\documentclass[submission,copyright,creativecommons]{eptcs}
\providecommand{\event}{Gandalf 2022} 
\usepackage[english]{babel}

\usepackage{lineno,hyperref}
\modulolinenumbers[5]
\usepackage{amsmath}
\usepackage{amssymb}
\usepackage{stmaryrd}
\usepackage{extarrows}
\usepackage{graphicx}
\usepackage{caption}
\usepackage[linguistics]{forest}

\usepackage{xspace}
\usepackage{url}
\usepackage{proof}
\usepackage{pgf}
\usepackage{tikz}
\usepackage{verbatim}

\newcommand\OldLong{old-long}

\usepackage{enumerate}
\renewcommand\delta{{\Delta'}}

 \newcounter{theo}
\usepackage[procnames]{listings}
\lstdefinelanguage{sml}{
	        morekeywords={at,of,case,letrec,thread,in,to,fun,proc,where,raise,exception,end,function,let,for,while,do,or,
                and,add,remove,match,if,then,such,that,else,skip,exit,not,inner,loop,return,repeat,until,input,output, try, catch, when, is, While},
	        keywordstyle={\sffamily\bfseries},
	        morekeywords=[2]{Algorithm,Parameter,Input,Output},
	        keywordstyle=[2]{\scshape},
	        morecomment=[l]{//},
	        morecomment=[s]{/*}{*/},
	        breaklines=true
}
\lstnewenvironment{sml}{\lstset{language=sml, numbers=left, numberstyle=\tiny, escapeinside={(*@}{@*)} }}{}

\newcommand\QNN[1]{\texttt{QN#1}\xspace}

\input{\OldLong /core/not-final-version.sty}
\newcommand\comments[2]{}
\newcommand\joachimcolor[1]{#1}

\newcommand\Iovkapc[1]{}

\newcommand\TBD[1]{}

\newcommand\antonio[1]{}

\newcommand\Short[1]{#1}
\newcommand\Long[1]{}
\newcommand\INLINE[1]{}
\usepackage[appendix=strip]{apxproof}

\input{core/macros-eptcs.sty}
\newtheoremrep{theorem}{Theorem}
\newtheoremrep{lemma}[theorem]{Lemma}
\newtheoremrep{proposition}[theorem]{Proposition}
\newtheoremrep{corollary}[theorem]{Corollary}
\newtheoremrep{definition}[theorem]{Definition}
\newtheoremrep{claim}[theorem]{Claim}
\newtheoremrep{example}[theorem]{Example}

\newcommand\imagepath{\OldLong /images/}

\begin{document}
\title{Schema-Based Automata Determinization}
\def\titlerunning{Schema-Based Determinization}
\def\authorrunning{Niehren, Sakho, and Al Serhali}

\author{Joachim Niehren 
\institute{ Inria, France \quad Universit\'{e} de Lille} 
\email{joachim.niehren@inria.fr}
\and
Momar Sakho
\institute{ Inria, France \quad  Universit\'{e} de Lille}
\email{momar.sakho@inria.fr}
\and 
Antonio Al Serhali  
\institute{ Inria, France \quad  Universit\'{e} de Lille}
\email{antonio.al-serhali@inria.fr}}

\maketitle
\renewcommand\I[1]{\ensuremath{\mathcal{I}_{#1}}}
\newcommand\nodes{\ensuremath{\mathit{nod}}}
\newcommand\sha{\ensuremath{\mathit{sha}}}

\newcommand\xmlonex{xml\&\onext}
\newcommand\applyelse{\ensuremath{@_{\ELSE}}}
\newcommand\Applyelse{{@_{\_}}}
\newcommand\PARAM[1]{\textsc{[#1]}}
\newcommand\SSHAET{\SHAtext\PARAM{\ELSE:T,\applyelse}\xspace}
\newcommand\dSSHAET{\dSHAtext\PARAM{\ELSE:T,\applyelse}\xspace}
\newcommand\dSSHAETs{\dSHAstext\PARAM{\ELSE:T,\applyelse}\xspace}
\newcommand\dSSHAEIT{\dSHAtext\PARAM{\ELSE:\I{T},\applyelse}\xspace}
\newcommand\dSSHAEITs{\dSHAstext\PARAM{\ELSE:\I{T},\applyelse}\xspace}
\newcommand\SSHAETs{\SHAstext\PARAM{\ELSE:T,\applyelse}\xspace}
\newcommand\SSHAEIT{\SHAtext\PARAM{\ELSE:\I{T},\applyelse}\xspace}
\newcommand\SSHAEITs{\SHAstext\PARAM{\ELSE:\I{T},\applyelse}\xspace}
\newcommand\SSHAE{\SHAtext\PARAM{\ELSE,\applyelse}\xspace}
\newcommand\dSSHAE{\dSHAtext\PARAM{\ELSE,\applyelse}\xspace}
\newcommand\SHAE{\SHAtext\PARAM{\applyelse}\xspace}
\newcommand\dSHAE{\dSHAtext\PARAM{\applyelse}\xspace}
\newcommand\SSHAEs{\SHAtext\PARAM{\ELSE,\applyelse}\xspace}

\newcommand\dSSHA{\dSHAtext\PARAM\ELSE\xspace}
\newcommand\dSSHAs{\dSHAstext\PARAM\ELSE\xspace}
\newcommand\SSHA{\dSHAtext\PARAM\ELSE\xspace}
\newcommand\SSHAs{\dSHAstext\PARAM\ELSE\xspace}
\newcommand\ESHA{\SHAtext\PARAM\ELSE\xspace}
\newcommand\ESHAs{\SHAstext\PARAM\ELSE\xspace}

\newcommand\SDFA{\DFAtext\PARAM\ELSE\xspace}
\newcommand\SNFAs{\NFAtext{s}\PARAM\ELSE\xspace}
\newcommand\SDFAs{\DFAstext\PARAM\ELSE\xspace}
\newcommand\successor{\mathit{next}}
\renewcommand\DD{{:}{:}}
\newcommand\Eval[3]{\generalevalcand{#3}{#2}{#1}}
\newcommand\seval[3]{\Eval{#1}{\{#2\}}{#3}}
\newcommand\TIR{\OP\hspace{-.1em}\CL}
\newcommand\accdet[1]{\states^{\det(#1)}}%
\newcommand\detp{\det'}
\newcommand\accdetp[1]{\states^{\detp(#1)}}%
\providecommand\rightarrowdashed{\relbar\rightarrow}

\newcommand\treetrans[2]{\ensuremath{#1 \rightarrowdashed #2}}
\newcommand\NOIA{\text{ has no internal $a$-rule}}
\newcommand\detS{\DET_S}
\newcommand\detgen[1]{\DET_{#1}}
\newcommand\detSp{\detp_S}
\newcommand\simp{\sim'}
\newcommand\rules{\Delta}
\newcommand\A{A}
\renewcommand\xtree{\mathit{x}\textrm{-}\mathit{tree}}
\newcommand\xdocument{\mathit{x}\textrm{-}\mathit{document}}

\newcommand\Qinith{\Qinit}
\newcommand\Qinitt{\Qinittgen\Delta}
\newcommand\Qinittgen[1]{\TIR^{#1}}
\newcommand\Qinittgendotted[1]{\dot{\TIR}^{#1}}
\newcommand\dottimes{\dot{\times}}
\newcommand\hattimes{\widehat{\times}}

\renewcommand\SAmathLongU{{\ensuremath{\SAmath=(\Sigma,\states,\Delta,\Qinith,\Qfin)}}}

\newcommand\dotteddelta{\dot{\Delta}}
\newcommand\dottedstates{\ensuremath{\mathcal{\dot{Q}}}}
\newcommand\epsclosure{\ensuremath{\epsilon^{\Delta^*}}}
\renewcommand\states{\ensuremath{\mathcal{Q}}}

\newcommand\elser[1]{\ensuremath{\QT_{\mathit{else}}^{#1}}}
\newcommand\QH{\ensuremath{\states_{\mathit{h}}}}
\newcommand\QT{\ensuremath{\states_{\mathit{t}}}}
\newcommand\Qsink{\ensuremath{\states_{\mathit{sink}}}}
\newcommand\sink{\ensuremath{\mathit{sink}}}
\newcommand\projectfirstname{\operatorname{\Pi}}
\newcommand\projectfirst[2]{\ensuremath{\projectfirstname_{#2}(#1)}}
\newcommand\Projectfirst[1]{\ensuremath{\projectfirstname_{#1}}}

\newcommand\prf[1]{\ensuremath{\Pi}_{#1}}

\newcommand\dotprf[1]{\dot{\projectfirstname}_{#1}}

\newcommand\projection{\operatorname{projectF}}
\newcommand\firstprojection[2]{\ensuremath{\projection_{#2}(#1)}}
\newcommand\Firstprojection[1]{\ensuremath{\projection_{#1}}}
\newcommand\projectF[1]{\Firstprojection{#1}}

\newcommand\scL{\scl{S}(A)}

\newcommand\SCL{\hatscl{S}(A)}
\newcommand\DSCL{\hatscl{S}(det(A))}

\newcommand\project{\mathit{project}}
\renewcommand\operatorname[1]{\mathit{#1}}
\newcommand\schemacleanname{\operatorname{scl}}
\newcommand\coaccclean{\operatorname{coacl}}
\newcommand\accclean{\operatorname{acl}}
\newcommand\schemaclean[2]{\ensuremath{\schemacleanname_{#2}(#1)}}
\newcommand\schemacleanp[2]{\ensuremath{\schemacleanname'_{#2}(#1)}}
\newcommand\Schemaclean[1]{\ensuremath{\schemacleanname_{#1}}}
\newcommand\dotscl[1]{\ensuremath{\dot\schemacleanname_{#1}}}
\newcommand\hatscl[1]{\ensuremath{\widehat\schemacleanname_{#1}}}
\newcommand\scl[1]{\Schemaclean{#1}}
\newcommand\sclS{\scl{S}} 
\newcommand\projS{\pi_S}
\newcommand\ins[1]{\mathit{\iota}_{#1}}
\newcommand\insS{\ins{S}} 

\newcommand\sclv[1]{\ins{#1}}
\newcommand\sclcS{\insS}
\newcommand\sclvS{\insS} 
\newcommand\insSs[1]{\insS^{#1}} 

\newcommand\Infer[2]{\infer{#2}{#1}}
\newcommand\complete{\mathit{cpl}}
\newcommand\AtimescplS{\mathit{A\times \complete(S)}}

\newcommand\IhstatePic[1]{\xrightarrow{}\hspace{-.4em}\tikz[baseline=(char.base)]{
            \node[shape=circle,draw,inner sep=0,outer sep=0] (char) {\ensuremath{#1}};}}
\newcommand\ItstatePic[1]{\xrightarrow{\langle\rangle}\hspace{-.4em}\tikz[baseline=(char.base)]{
    \node[shape=circle,draw,inner sep=0,outer sep=0] (char) {\ensuremath{#1}};}}

\newcommand\statePic[1]{
\tikz[baseline=(char.base)]{            \node[shape=circle,draw,inner sep=0,outer sep=0] (char) {\ensuremath{#1}};}}

\newcommand\remove[1]{\mathit{rm}_{#1}}
\renewcommand\complement[1]{\overline{#1}}
\newcommand\PAS{A\times S}
\newcommand\WRT{ \ \text{wrt.}\ }
\newcommand\MULTC{\ }


\renewcommand\Figure[1]{Figure~\ref{#1}}
\begin{abstract}
We propose an algorithm for schema-based determinization of finite
automata on words and of stepwise hedge automata on nested words.
The idea is to integrate schema-based cleaning directly into automata
determinization. We prove the
correctness of our new algorithm and show that it is always more
efficient than standard determinization followed by schema-based
cleaning. Our implementation permits to obtain a small deterministic automaton
for an example of an XPath query, where standard determinization yields
a  huge stepwise hedge automaton for which schema-based
cleaning runs out of memory.

\end{abstract}

\section{Introduction}

Nested words are words enhanced with well-nested parenthesis. They generalize over
trees, unranked trees, and sequences of thereof that are also called
hedges or forests. Nested words provide a formal way to
represent semi-structured textual documents of XML and JSON format.

Regular queries for nested words can be defined by finite state
automata. We will use stepwise hedge automata (\SHAs) for this purpose \cite{Sakho},
which combine finite state automata for words and trees in a natural manner. \SHAs
refine previous notions of hedge automata from the sixties
\cite{Thatcher67b,TATA07} in order to obtain a
decent notion of left-to-right and bottom-up determinism.
They extend on stepwise tree automata \cite{CarmeNiehrenTommasi04}, so
that they can not only be applied to unranked trees but also to
hedges. Any \SHA defines a forest algebra \cite{DBLP:conf/birthday/BojanczykW08}
based on its transition relation.
Furthermore, \SHAs can always be  determinized and have the
same expressiveness as (deterministic) nested word automaton (\NWA) \cite{DBLP:conf/icalp/Mehlhorn80,BRAUNMUHL19851,Alur07,DBLP:journals/sigact/OkhotinS14}. Note, however, that \SHAs do not provide any form of
top-down determinism in contrast to \NWAs.

Efficient compilers from regular \XPath queries to \SHAs exist \cite{Sakho},
possibly using \NWAs as intermediates \cite{AlurMadhusudan09,MozafariZengZaniolo12,debarbieux15}.
Our main motivation is to determinize the \SHAs of  regular \XPath
queries since deterministic automata are crucial for various
algorithmic querying tasks. In particular, determinism
reduce the complexity of universality or inclusion checking
from EXP-completeness to P-time, both for the classes of
deterministic \SHAs or \NWAs. In turn, universality checking
is relevant for the earliest query answering of XPath queries
on XML streams \cite{GauwinNiehrenTison09b}. Furthermore, determinism
is needed for efficient in-memory answer enumeration of regular queries \cite{schmid_et_al:LIPIcs.ICDT.2021.4}. 

Automata determinization may take exponential time in the worst case, so
it may not always be feasible in practice. For \SHAs compiled from
the \XPath queries of the XPathMark benchmark \cite{FranceschetXPathPT},
however, it was shown to be unproblematic. This changes for the
\XPath benchmark collected by Lick and Schmitz \cite{lickbench}:
for 37\% of its regular \XPath queries, \SHA determinization does
require more than 100 seconds, in which case it produces
huge deterministic automata \cite{alserhalibench}. An example is:
\begin{verbatim}
      (QN7)     /a/b//(* | @* | comment() | text())
\end{verbatim}
This XPath query selects all nodes of an XML document that are descendants of
a $b$-element below an $a$-element  at the root. The nodes may
have any XML type: element, attribute, comment, or
text. The nondeterministic \SHA for \QNN7 has $145$ states and an overall size of $348$. Its determinization however leads to an automaton with $10.005$ states and an overall size of $1.634.122$. 

A kick-off question is how to reduce the size
of deterministic automata. One approach beside of minimization
is to apply schema-based cleaning \cite{Sakho}, where the
schema of a query defines to which nested words the query can be applied.
Schemas are always given by deterministic automata while the automata for
queries may be nondeterministic. The idea of schema-based automaton cleaning is to keep only those states and transition rules of the automaton, that are needed to
recognize some nested word satisfying the schema. The needed states and rules
can be found by building the product of automata for the query and the schema.
For \XPath queries selecting nodes, we have the schema $\onext$ that states
that a single node is selected for a fixed variable $x$ by any answer of the query.
The second schema expresses which nested words satisfy the XML data model.
With the intersection of these two schemas, the schema-based cleaning of the deterministic \SHA for $\QNN7$ indeed has only $74$ states and $203$ rules. When applying \SHA minimization afterwards, the size of the automaton goes down to $27$ states and $71$ transition rules. However, our implementation of schema-based cleaning, runs out of memory for larger automata with say more than $1000$ states.
Therefore, we cannot compute the schema-based cleaning from the
deterministic \SHA obtained from $\QNN7$. Neither can we minimize
it with our implementation of deterministic \SHA minimization.
The question of how to produce small deterministic
automaton for queries as simple as $\QNN7$ 
thus needs a different answer.

Given the relevance of schemas, one naive approach could be to
determinize the product of the automata for the query
and schema. This may look
questionable at first sight, given that the schema-product
may be bigger than the original automaton, so
why could it make determinization more efficient? But
in the case of $\QNN7$, the determinization of the
schema-product yields a deterministic
automata with only $92$ states and $325$ transition rules,
and can be computed efficiently. This observation is very promising,
motivating three general questions:
\begin{enumerate}
\item Why are schemas so important for automata determinization?
\item Can this be established by some complexity result?
\item Is there a way to compute the schema-based cleaning of
  the determinization of an \SHA more efficiently than be schema-based cleaning followed by determinization?
\end{enumerate}
Our main result is a novel algorithm for
schema-based determinization of \NFAs and \SHAs, that integrates
schema-based cleaning directly into the usual
determinization algorithm. This algorithm answers question 3 positively.
Its idea is to keep only  those subsets of states of the automaton
during the determinization, that can be aligned to some
state of the schema. In our Theorem \ref{detS2},  we prove that
schema-based determinization always produces the same
deterministic automaton than schema-free determinization followed 
by schema-based cleaning. By schema-based determinization
we could compute the schema-based
cleaning of the determinization of $\QNN7$ in less than three seconds. In
contrast, the schema-based cleaning of
the determinization does not terminate after a few hours. In the general case, 
the worst case complexity of schema-based determinization
is lower than schema-less determinization followed by
schema-based cleaning.

We also provide a more precise complexity upper bound in
Proposition \ref{prop:sdet-sha}. Given an nondeterministic
\SHA $A$ let $det(A)$ be its determinization, and given a
deterministic \SHA $S$ for the schema, let $A\times S$
the accessible part of the schema-product, and $\sclS(A)$
the schema-based cleaning of $S$ with respect to schema $S$.
We show that the upper bound for the maximal computation time of $\sclS(det(A))$
depends quadratically on the number of states of $S\times det(A)$, which is
often way smaller than for $det(A)$ since $S$ is deterministic.
This complexity result shows why the schema is so relevant
for determinization (questions 1 and 2), and why
computing the schema-based determinization is often more
efficient than determinization followed by schema-based
cleaning (question 3).

To see that $S\times det(A)$ is often way smaller than $det(A)$
for deterministic $S$ we first note that $det(A\times S)=det(A) \times S$ since $S$ is
  deterministic.\footnote{If $\{(q_1,s_1)\ldots (q_n,s_n)\} \in
det(A\times S)$ then there exists a tree that can go into all states
$q_1\ldots q_n$ with $A$ and into all states $s_1,\ldots s_n $ with $S$.
Since $S$ is deterministic, we have $s_1=\ldots s_n$. So there exists a 
tree going into $\{q_1,\ldots,q_n\}$ with $det(A)$ and also into all
$s_i$. So $(\{q_1,\ldots,q_n\},s_i)$ is a state of $det(A)\times S$.}
So for the many states $Q=\{q_1\ldots q_n\}$ of $det(A)$ there
may not exist any state $s$ of $S$ such that $(Q,s) \in det(A)\times
S$, because this requires all states $q_i$ can be aligned to $s$,
i.e. that $(q_i,s)$ in $A\times S$ for all $1\le i\le n$.
Furthermore, $\det(A)\times S$ is equal to $\detS(A)\times S$, so that
$det(A\times S)=\detS(A)\times S$. Hence any size bound for
the schema-based determinization $\detS(A)$ implies a size bound
for the determinization of the schema-product. Also,
in our experiments $\detS(A)\times S$ is almost by a factor of $2$
bigger than $\detS(A)$.  So the size of the determinization of the
schema-product is closely tied to the size of the
schema-based determinization. 

We also present a experimental evaluation of our implementation of
schema-based determinization of \SHAs. We consider a
scalable family of \SHAs obtained from a scalable family of
\XPath queries. Our experiments confirm the very large
improvement implied by the usage of schemas for determinization. For this,
we implemented the algorithm for schema-based \SHA
determinization in Scala. Furthermore, we applied the
XSLT compiler from regular forward \XPath queries to
\SHAs from \cite{Sakho}, as well as the datalog
implementations of \SHA minimization and
schema-based cleaning from there.

A large scale experiment on practical \XPath queries was
provided in follow-up work \cite{alserhalibench} where schema-based
algorithms were applied to the regular \XPath queries
collected by Lick and Schmitz \cite{lickbench} from real word XQuery
and XSLT programs. Small deterministic \SHAs could be obtained by
schema-based determinization for all
regular \XPath queries in this corpus. In contrast, standard
determinization in 37\% of the cases fails with a timeout of
100 seconds. Without this timeout, determinization
either runs out of memory or produces very large automata.

\medskip
\paragraph{Outline}
We start with related work on automata for nested words,
determinization for \XPath queries (Section \ref{sec:rel}).
In Section \ref{sec:words}, we recall the definition \NFAs and discuss how to use  them as schemas and queries on words. In Section
  \ref{sec:scl-words}, we recall schema-based cleaning for \NFAs.
  In Section \ref{sec:schema-det-words}, we contribute
  our schema-based determinization algorithm in the case of \NFAs and show
  its correctness. In Section \ref{sec:nw}, we recall the notion of \SHAs
  for defining languages of nested words. In  Section
  \ref{sec:schema-det}, we lift schema-based determinization to
  \SHAs. \Short{Full proofs can be found in the Appendix of the long version
  \cite{niehrenschemadet}.} \Long{Full proofs can be found in the Appendix.}
\section{Related Work}
\label{sec:rel}

We focus on automata for nested words, even
though our results are new for \NFAs too.

\paragraph{Nested word automata}
As recalled in the survey of Okhotin and
Salomaa \cite{DBLP:journals/sigact/OkhotinS14}, Alur's et al. \cite{Alur07}
\NWAs were first introduced in the eighties under the name of input driven automata by Mehlhorn \cite{DBLP:conf/icalp/Mehlhorn80}, and then reinvented
several times under different names. In particular, they
were called visibly pushdown automata \cite{AlurMadhusudan04}, pushdown forest automata \cite{NeumannSeidl98}, and streaming tree automata \cite{GauwinNiehrenRoos08}.
The determinization algorithm for \NWAs was first
invented in the eighties by von Braunmühl and Verbeek
in the journal version of \cite{BRAUNMUHL19851} and then rediscovered
various times later on too.

\paragraph{Determinization algorithms}
The usual determinization algorithms for \NFAs relies on
the well-known subset construction. The determinization
algorithms of bottom-up tree automata and \SHAs are
straightforward extensions thereof.
The determinization algorithm for \NWAs, in contrast,
is more complicated, since having to deal with pushdowns.
Subsets of \emph{pairs} of states are to be considered there
and not only subsets of states as with the usual automata
determinization algorithm. We also notice that general
pushdown automata with nonvisible stacks can even not always be determinized.

\paragraph{Application to  \XPath}
Debarbieux et al. \cite{debarbieux15}  noticed that the
determinization algorithm for \NWAs often behaves badly when applied
to \NWAs obtained from XPath queries as simple as $//a/b$. Niehren and 
Sakho \cite{Sakho} observed more recently
that the situation is different for the  determinization of \SHAs: It works out nicely for the \SHA of $//a/b$ and also for all other \SHAs obtained by compilation from forward
navigational XPath queries in the XPathMark benchmark
\cite{FranceschetXPathPT}. 
Even more surprisingly, the same good
behavior could be observed for
the determinization algorithm of \NWA when restricted to
\NWAs with the weak-single entry property.

\paragraph{Weak single-entry \NWAs versus \SHAs}
The weak-single entry property
implies that an \NFA cannot memoize anything in its
state when moving top-down. So it can only pass information
left-to-right and and bottom-up, similarly to an \SHA.
This property
failed for the \NWAs considered by Debarbieux et al. and 
the
determinization of their \NWAs thus required top-down determinization.
This quickly led to the size explosion described above.
One the other hand side, the weak single-entry
property can always
be established in quadratic time by compiling \NWAs to \SHAs
forth and back. Or else, one can avoid top-down determinization all
over by directly working with \SHAs as we do here.

\ignore{\paragraph{Forest Algebras}
Nested words can be seen as forests, i.e., sequences of
unranked trees. There may be minor differences in the treatment
of labels, which appear in tree constructors, as labels of
parenthesis, or of internal letters of the nested words.
Up to encodings that remove such differences, any \SHA defines a forest algebra
based on its transition relation \cite{DBLP:conf/birthday/BojanczykW08}.}

\ignore{\paragraph{Nested regular path queries} 
\cite{LibkinMartensVrgoc13,martens_et_al:LIPIcs:2018:8594} are
formulas from propositional dynamic logic (PDL)
\cite{DBLP:journals/jcss/FischerL79}. While applicable
to general data graphs, they can also be restricted to
nested words. They extend on usual regular expression by adding 
nested filters that are closed under the logical operators. 
Filters may test for the existence of nodes answering a
nested regular path query. When XML documents,
nested regular path queries can be identified with
regular forward XPath queries
\cite{arenas2011querying,GottlobKochPichler03c}.
t is folklore that nested regular path queries on data trees can be
compiled to automata.

\paragraph{MSO}
For ranked trees, one can first compile path queries to 
the monadic second-order (MSO) formulas, and from there
to tree automata that recognize so-called V-structures
\cite{ThatcherWright68,GottlobKoch02,CarmeNiehrenTommasi04}. 
V-structures are trees that are annotated with variables $x$
satisfying
$\onex$, saying
where the any variable assignement has to assign a single
node to variable $x$. In recent database
terminology, languages of V-structures are called document 
spanners
\cite{DBLP:journals/jacm/FaginKRV15,schmid_et_al:LIPIcs.ICDT.2021.4}. 
But since the compilation of path queries to MSO formulas eagerly introduces 
quantifier alternations, that are to be eliminated 
by repeated automata determinization, 
this approach leads to a large size blowup. 

\paragraph{Nested Regular Expressions}
A more recent idea \cite{Sakho} is
to compile nested regular path queries to nested regular
expressions in a first step, i.e., regular expressions for
nested words that were introduced earlier under the name
\emph{regular expression types} by Hosoya and Pierce \cite{DBLP:journals/toit/HosoyaP03}. Nested regular expressions support  the
usual operators of regular expressions and the nesting
expressions $\OP e \CL$ for defining languages
of nested words. Furthermore, they support
vertically recursive definitions based
on $\mu$-expressions $\mu x.e$,
intersections $e\cap e'$,
and complementation $\overline{e}$. It is then
possible to compile nested regular expressions to \SHAs or \NWAs,
by lifting the usual automata constructions from standard regular to
nested regular expressions. }

\section{Finite Automata on Words, Schemas, and Queries}
\label{sec:words}
In this section, we discuss hwo to use \NFAs 
for defining schemas and queries on words.

Let $\mathbb{N}$ be the set of natural numbers including $0$.
The set of words over a finite alphabet $\Sigma$
is $\Sigma^*=\bigcup\limits_{n \in \mathbb{N}}^\infty \Sigma^n$.
A word $(a_1,\ldots, a_n)\in\Sigma^n$ is written
as $a_1\ldots a_n$. We denote by $\epsilon$ the empty word,
i.e., the unique element of $\Sigma^0$ and by $w_1\cdot w_2\in \Sigma^*$
the concatenation of two words $w_1,w_2\in \Sigma^*$.
For example, if $\Sigma=\{a,b\}$ then $aa\cdot bb=aabb=a\cdot a\cdot b \cdot b$.

\begin{definition}
A \NFA  is a tuple $\SAmath=(\Sigma,\states,\Delta,\Qinith,\Qfin)$ such that
$\states$ is a finite set of states, the alphabet $\Sigma$ is a finite set,
$\Qinith, \Qfin \subseteq \states$ are subsets of initial and final
states, and $\Delta \subseteq \states \times \Sigma\times \states$ is
the set of transition rules.
\end{definition}
The size of a \NFA is $|A|=|\states|+|\Delta|$. 
A transition rule $(q,a,q')\in\Delta$ is 
denoted by $q \xrightarrow{a} q'\in\Delta$. We define transitions $q \xrightarrow{w} q'\wrt\Delta$ for arbitrary words $w\in\Sigma^*$ by the following inference rules:
$$
\infer{q\xrightarrow{\epsilon}q\wrt \Delta}
      {q\in \states}
      \qquad
\infer{q\xrightarrow{a}q'\wrt \Delta}
      {q\xrightarrow{a} q'\in \Delta}
      \qquad
\infer{q_0\xrightarrow{w_1\cdot w_2}q_2 \wrt \Delta}
      {q_0\xrightarrow{w_1}q_1 \wrt \Delta\qquad
       q_1\xrightarrow{w_2}q_2 \wrt \Delta }
$$
      The language of words recognized by a \NFA then is
      $\L(A)=\{w \in\Sigma^*\mid  q\xrightarrow{w} q' \wrt \Delta,\ q\in \Qinit,\ q'\in \Qfin\}$.

\begin{figure}[t]
$$
  \infer{
  \Qinith^A\in \Qinit^{det(A)} \qquad \Qinith^A\in \accdet\A 
  }
  {
   \Qinith^A\not=\emptyset
  }
  \qquad
\infer[]
         {Q\in \FIN^{\det(A)}}
         {Q\in \states^{\det(A)} \quad Q \cap \FIN^{A}  \not=\emptyset }
$$
$$
   \infer{
  \SAir{Q}{a}{Q'} \in \Delta^{det(A)} \qquad Q' \in \accdet\A
  }
  {
\begin{array}{c}
  Q \in \accdet\A
\qquad
  Q'=\{q' \in \states^A \mid \SAir{q}{a}{q'} \in\Delta^A,\ q \in Q \} \not=\emptyset 
  \end{array}
  }
$$
$$
  \det(A)=(\Sigma,\states^{\det(A)},\Delta^{\det(A)}, \Qinit^{\det(A)},\Qfin^{\det(A)})
$$
  \caption{\label{acc-det} \label{accdet} The accessible determinization
    $\det(A)$ of \NFA $A$.
    }
\end{figure}

\begin{figure}[t]
   \begin{sml}
fun det($A$) =
  let $Store=hashset.new(\emptyset)$ and $Agenda=list.new()$ and $Rules=hashset.new(\emptyset)$
  if $\init^A\not=\emptyset$ then $Agenda.add(\init^A)$   
  while $Agenda.notEmpty()$ do 
    let $Q = Agenda.pop()$
    let $h$ be an empty hash table with keys from $\Sigma$.
       // the values will be nonempty hash subsets of $\states^A$
    for $q\xrightarrow{a} q' \in \Delta^A$ such that $q\in Q$ do
        if $h.get(a)=undef$ then $h.add(a,hashset.new(\emptyset))$
        $(h.get(a)).add(q')$ 
    for $(a,Q')$ in $h.tolist()$ do $Rules.add(\SAir{Q}{a}{Q'})$
        if not $Store.member(Q')$  then $Store.add(Q')$ $Agenda.push(Q')$
  let $\init^{det(A)}=\{Q \mid Q \in Store, Q \cap \init^A \not= \emptyset\}$ and $F^{det(A)}=\{Q \mid Q \in Store, Q \cap \Qfin^A \not= \emptyset\}$
  return $(\Sigma,Store.toSet(),Rules.toSet(),\init^{det(A)},F^{det(A)})$
\end{sml}
  \caption{\label{algo-accdet-words}  A program computing the accessible determinization of an \NFA $A$ from \Figure{acc-det}.}
 \end{figure}

A \NFA $A$ is called \emph{deterministic} or equivalently a \DFA, if it has
at most one initial state, and for every pair $(q,a) \in \states \times
\Sigma$ there is at most one state $q'\in\states^A$ such that $q \xrightarrow{a} q' \in \Delta^A$.
Any \NFA $A$ can be converted into a \DFA  that
recognizes the same language by the usual subset
construction. The accessible determinization $\det(A)$ of $A=(\Sigma,\states^A,\Delta^A,\Qinit^A,\Qfin^A)$ is defined by the inference rules in \Figure{acc-det}. It works like the usual subset construction, except that only accessible subsets are created. It is well known that $\L(A)=\L(\det(A))$. Since only accessible subsets of states are added, we
have $\states^{\det(A)}\subseteq 2^{\states^A}$. Therefore,
the accessible determinization may even reduce the size
of the automaton and often avoid the 
exponential worst case where $\states^{\det(A)}=2^{\states^A}$.

\begin{proposition}[Folklore]\label{comp-detacc-words}
  The accessible determinization $\det(A)$ of a \NFA $A$ can be computed
  in expected amortized time
  $O(|\states^{\det(A)}|\MULTC| \Delta^A| + |A|).$
\end{proposition}
\begin{proofsketch}
The algorithm for accessible determinization with this complexity is somehow folklore. We
 sketch it nevertheless, since we need to refined it  for
   schema-based determinization later on.
A set of inference rules for accessible determinization is given in
\Figure{acc-det}, and an algorithm computing the fixed point of these
inference rules is presented in \Figure{algo-accdet-words}. It uses dynamic perfect
   hashing \cite{DBLP:journals/siamcomp/DietzfelbingerKMHRT94} for
   implementing hash sets, so that set inserting and membership can be done in
   randomized amortized time $O(1)$. The algorithm has a hash set $Store$
   to save all discovered states $\states^{det(A)}$ and a hash set $Rules$
   to collect all transition rules. Furthermore, it has a stack
   $Agenda$ to process all new states $Q \in
   \states^{det(A)}$. \ignore{For each $Q$ popped from
   the stack $Agenda$, the algorithm uses a hash table $h$ to compute
   all pairs $(a,Q')$ such that $Q\xrightarrow{a} Q'\in\Delta^{\det(A)}$
   and $Q'\not=\emptyset$. This is done by iterating of $\Delta^A$
   so in time $O(|\Delta^A|)$.
   By iterating over the hash table $h$, all transitions $\SAir{Q}{a}{Q'}$
   will be added to the set $Rules$ and $Q'$ will be added to the
   stack $Agenda$ and to the hash set $Store$ if it wasn't there yet. The
   overall number of elements in the $Agenda$ is
   $|\states^{det(A)}|$. For each $Q$, the computation of all $Q'$ is
   in time $O(|\Delta^A|)$. The preprocessing of $A$ requires time
   $O(|A|)$. Thus, the total time of the algorithm is in
   $O(|\states^{\det(A)}|\MULTC| \Delta^A| + |A|)$.}
\end{proofsketch}
\begin{figure}[t]
\begin{minipage}{\textwidth}  
\begin{minipage}{.49\textwidth}  
\centering
\includegraphics[scale=.29]{\imagepath p4a/scp4a.png}
\caption{\label{scp4a}The \NFA $A_0$ for the regular expression $(x + \epsilon).(x.a)^*$}
\end{minipage}
\quad
\begin{minipage}{.49\textwidth}  
\includegraphics[scale=.29]{\imagepath p4a/detscp4a.png}
\caption{\label{detscp4a}The accessible determinization $\det(\NFAEX)$ up to the renaming
  of states $[\{2,4\}/0,$$\{2,3\}/1,$$\{2\}/2,\{3\}/3 ]$.}
\end{minipage}
\end{minipage}
\end{figure}
As a running example, we consider the \NFA $A_0$ for the regular expression $(x + \epsilon).(x.a)^*$ that is drawn as a labeled digraph in \Figure{scp4a}: the nodes of the graph are the states and the labeled edges represent the transitions rules. The initial states are indicated by an ingoing arrow and the final state are doubly circled.
The graph of the \DFA $\det(A_0)$ obtained by accessible determinization
is shown in \Figure {detscp4a}. It is given up to a renaming of the
states that is given in the caption. Note that only $4$ out of the $2^3=8$
subsets are accessible, so the size increases only by a single state and two
transitions rules in this example.

\label{sec:schema-queries}
A \emph{regular schema} over $\Sigma$ is
  a \DFA with the alphabet $\Sigma$. We next show how to
  use automata to define 
\\
\noindent  
\begin{minipage}{\textwidth}
\begin{minipage}{.5\textwidth}
regular queries
  on words. For this, any
word is seen   as a labeled digraph. The labeled   digraph of   the word $aab$, for instance, is drawn to the right. The

  \end{minipage}
\begin{minipage}{.5\textwidth}
\centering \includegraphics[scale=.4]{\imagepath aab.png}
\end{minipage}
\end{minipage}
\noindent   set of nodes of    the graph  is the set
of positions of the word $\POS(w)=\{0,\ldots,n\}$ where $n$
is the length of $w$. Position $0$ is labeled by $start$, while all
other positions are labeled by a single letter in $\Sigma$.
A monadic query function on words with alphabet $\Sigma$ is a
total function $\QUERY$ that
maps some words $w\in \Sigma^*$ to a subset of position $\QUERY(w)\subseteq \POS(w)$.
We say that a position $\pi\in\POS (w)$ is selected by $\QUERY$ if $w\in\dom(\QUERY)$
and $\pi\in \QUERY(w)$. 

Let us fix a single variable $x$. Given a position $\pi$ of a word
$w\in\Sigma^*$ let $w*[\pi/x]$ be the word obtained from $w$ by
inserting $x$ after position $\pi$. We note that all words
of the form $w*[\pi/x]$ contain a single occurrence of $x$. 
Such words are also called $V$-structures where $V=\{x\}$
(see e.g \cite{straubing1994finite}).

\begin{figure}[t]
  \noindent
  \begin{minipage}{.5\textwidth}
\centering
\includegraphics[scale=.40]{\imagepath p4a/scdetscp4a.png}
\caption{\label{scldetA0}The schema-based cleaning of $\det(\NFAEX)$ with schema $\onex$.}
  \end{minipage}
  \qquad
  \begin{minipage}{.4\textwidth}
\centering
\includegraphics[scale=.30]{\imagepath words-one-x/det-stepwise-manual.png}
\caption{\label{DFAONEX}Schema $\onex$ with alphabet $\Sigma\uplus\{x\}$.}
   \end{minipage}
\end{figure}

The set of all $V$-structures can be defined by
the schema \onex over $\Sigma\uplus\{x\}$ in \Figure{DFAONEX}.
It is natural to identify any total monadic query function $\QUERY$
with the language of $V$-structures
$L_\QUERY=\{w*[\pi/x]\mid w\in\Sigma^*, \pi\in \QUERY(w)\}$.
This view permits us to define a subclass of total monadic query functions by automata.
A \emph{(monadic) query automaton over $\Sigma$} is a \NFA $A$ with alphabet
$\Sigma\uplus\{x\}$. It defines the unique
total monadic query function $\QUERY$ such that $L_{\QUERY}=\L(A)\cap \L(\onex)$.
A position $\pi$ of a word $w\in\Sigma^*$ is thus selected by the query $\QUERY$ on $w$
if and only if the $V$-structure $w*[\pi/x]$ is recognized by $A$, i.e.:
$$
   \pi \in \QUERY(w) \Leftrightarrow w*[\pi/x]\in \L(A)
$$
A query function is called regular if it can be defined by some \NFA.
It is well-known from the work of Büchi in the sixties \cite{Buchi60}
that the same class of regular query functions can be defined
equivalently by monadic second-order logic.

We note that only the words satisfying the schema $\onex$ (the $V$-structures) are
relevant for the query function $\QUERY$ of a query automaton $A$.
The query automaton \NFAEX in \Figure{scp4a} for instance,
defines the query function that selects the start position of the
words $\epsilon$ and $a$ and no other positions elsewhere. This is
since the subset of $V$-structures recognized by \NFAEX
is $x + x.a$. Note that the words $\epsilon$ and $xxa$
do also belong to $\L(\NFAEX)$, but are not $V$-structures,
and thus are irrelevant for the query function $\QUERY$.
\section{Schema-Based Cleaning}\label{sec:scl-words}
Schema-based cleaning was introduced only recently \cite{Sakho} in order to reduce the size of automata on nested words.
The idea is to remove all rules and states
from an automaton that are not used to recognize any word satisfying
the schema. Schema-based cleaning can be based on the accessible states of
the product
of the automaton with the schema. While this product may be larger than the
automaton, the schema-based cleaning will always be smaller.

For illustration, the schema-based cleaning of
\NFA $\det(\NFAEX)$ in \Figure{detscp4a} with respect
to schema $\onex$ is given in \Figure{scldetA0}.
The only words recognized by both $\det(\NFAEX)$
and $\onex$ are $x$ and $xa$. For recognizing these
two words, the automaton $\det(\NFAEX)$ does not need states $2$ and $3$, so
they can be removed with all their transitions rules.
Thereby, the word $xxa$ violating the schema is no more
recognized after schema-based cleaning, while it was
recognized by $\det(\NFAEX)$. 
Furthermore, note 
that the state $0$ needs no more to be final after schema-based cleaning.
Therefore the word $\epsilon$, which is recognized by the
automaton but not by the schema, is no more recognized after
schema-based cleaning. So schema-based
cleaning may change the language of the automaton
but only outside of the schema.

Interestingly, the $\NFA$ \NFAEX in \Figure{scp4a} is schema-clean
for schema $\onex$ too, even though it is not perfect, in that it recognizes the words
$\epsilon$ and $xxa$ which are rejected by the schema.
The reason is that for recognizing the words $x$ and $xa$,
which both satisfy the schema, all 3 states and
all 4 transition rules of \NFAEX are needed.
In contrast, we already
noticed that the accessible determinization $\det(\NFAEX)$
in \Figure{detscp4a} is not
schema-clean for schema $\onex$. This illustrates that accessible
determinization does not always preserve schema-cleanliness.
In other words, schema-based cleaning may have a stronger cleaning effect
after determinization than before.

\begin{figure}[t]
$$
\begin{array}{c}
  \infer[]
  {
  (q,s) \in \Qinith^{A\times S}
\qquad
  (q,s) \in \states^{A\times S}
  }
  {
  q \in \Qinith^A & s \in \Qinith^S
  }
   \qquad
  \infer[]
  {
  (q,s) \in \Qfin^{A\times S}
  }
  {
  q \in \Qfin^A & s \in \Qfin^S &  (q,s) \in \states^{A\times S}
  }
  \\
  \\
  \infer[]{
  \SAir{(q_1,s_1)}{a}{(q_2,s_2)} \in \Delta^{A\times S} \quad (q_2,s_2) \in \states^{A\times S}
  }
  {
   \SAir{q_1}{a}{q_2} \in \Delta^A & \SAir{s_1}{a}{s_2} \in \Delta^S \quad (q_1,s_1) \in \states^{A\times S}
  }

\end{array}
$$
\caption{\label{fig:product_sha_words}\label{accessible-product}
  Accessible product
  $A\times S=(\Sigma,\states^{A\times S},
\Qinith^{A\times S},  \Qfin^{A\times S},\Delta^{A\times S})
  $.}
\end{figure}

The schema-based cleaning of an automaton can be
defined based on the accessible product of
the automaton with the schema. The accessible product $A\times S$ 
of two \NFAs $A$ and $S$ with alphabet 
$\Sigma$ is defined in \Figure{accessible-product}. This
is the usual product, except that only accessible states
are admitted. Clearly, $\L(A\times S)=\L(A)\cap \L(S)$.
Let $\prf{A}(A \times S)$ be obtained from the accessible
product by projecting away the second component\Long{, as
formally defined in \Figure{prf-words} of the appendix}\INLINE{, as
formally defined in \Figure{prf-words}}. The schema-based cleaning of $A$ with
respect to schema $S$ is this projection.
\begin{definition}\label{sclpro-words}$\scl{S}(A) = \prf{A}(A \times S)$. \end{definition}
\begin{toappendix}
\begin{figure}[t]
$$
\begin{array}{c} 
  \infer[]
  {
  	q \in \Qinit^{\prf{A}(A \times S)}
  }
  {
       (q,s) \in \Qinit^{A\times S}
  }
  \qquad  
  \infer[]
  {
  	q \in \states^{\prf{A}(A \times S)} 
  }
  {
       (q,s) \in \states^{A\times S}
  }
\qquad
	\infer[]
         {q \in \Qfin^{\prf{A}(A\times S)}}
           {(q,s) \in \Qfin^{A\times S}}
\qquad
  \infer[]{
   \SAir{q_1}{a}{q_2} \in \Delta^{\prf{A}(A \times S)} 
  }
  {
  \SAir{(q_1,s_1)}{a}{(q_2,s_2)} \in \Delta^{A \times S} 
  }
\end{array}
$$
\caption{\label{prf-words}Projection  $
\prf{A}(A\times S) = (\Sigma, \states^{\prf{A}(A\times S)},\Delta^{\prf{A}(A\times S)},\Qinit^{\prf{A}(A\times S)},\Qfin^{\prf{A}(A\times S)})
$.}
\end{figure}
\end{toappendix}

The fact that $A\times S$ is restricted to accessible states
matches our intuition that all states of $\scl S (A)$ can be used
to read some word in $\L(A)$ that satisfies schema $S$. This can
be proven formally under the condition that all states of $A\times S$
are also co-accessible.
Clearly, $\scl{S}(A)$ is obtained from $A$ by
removing states, initial states, final
states, and transitions rules. So it 
is smaller or equal in size $|\scl S(A)|\le |A|$ and language
$\L(\scl{S}(A))\subseteq \L(A)$. Still,
schema-based cleaning preserves the language within the schema. 
\begin{proposition}[\cite{Sakho}]
$\L(A)\cap\L(S)=\L(\scl{S}(A)) \cap \L(S)$.
\end{proposition}

\begin{figure}[t]
  \noindent
  \begin{minipage}{.3\textwidth}
\centering
\includegraphics[scale=.25]{\imagepath scdetscReg2.png}
\caption{\label{scdetscReg2}A \DFA that is schema-clean 
  but not perfect for $\onex$.}
  \end{minipage}
  \quad
  \begin{minipage}{.3\textwidth}
\centering
\includegraphics[scale=.25]{\imagepath prodReg2OneX/orig/stepwise-manual.png}
    \caption{\label{scdetscReg2b}The accessible product
      with \onex is schema-clean 
  and perfect for $\onex$.}
   \end{minipage}
     \quad
   \begin{minipage}{.3\textwidth}
   \centering
\includegraphics[scale=.33]{\imagepath nwords-one-x/det-stepwise.png}
\vspace{-1cm}\caption{\label{one-x-step}
  The \dSHA $\onexnw$ with alphabet $\Sigma\uplus\{x,\neg x\}$.
}
   \end{minipage}
\end{figure}
Schema-clean deterministic automata may still not be perfect,
in that they may recognize some words outside the schema.
This happens for \DFAs if some state of is reached, both, by a word
satisfying the schema and another word that does not satisfy
the  schema.
An example for a \DFA that is schema-clean but not perfect for \onex is given in \Figure{scdetscReg2}. It is not perfect since it accepts the non $V$-structure
  $xaxa$. The problem is that state $1$ can be reached by the words $a$ and $xa$,
  so one cannot infer from being in state $1$ whether some $x$ was read or not.
  If one wants to avoid this, one can use
  the accessible product of the \DFA with the schema instead. In the
  example, this yields the \DFA in \Figure{scdetscReg2b} that is schema-clean and
  perfect for \onex.

\begin{proposition}[Folklore]\label{compl-scl-words}
  For any two \DFAs $A$ and $S$  with alphabet $\Sigma$ the accessible product $A\times S$ 
  can be computed in expected amortized time $O(|\states^{A\times S}| |\Sigma|  + |A| + |S| )$.
  \end{proposition}

\begin{proof}
An algorithm to compute the fixed points of the
inference rules for the accessible product $A\times S$
in \Figure{accessible-product} can be organized such that
only accessible states are considered (similarly to
semi-naive datalog evaluation). This
algorithm is presented in \Figure{algo-int-rules-words}.
  It dynamically generates the set of rules $Rules$ by using perfect dynamic hashing \cite{DBLP:journals/siamcomp/DietzfelbingerKMHRT94}.
Testing set membership is in time $O(1)$ and the addition
of elements to the set is in expected amortized time $O(1)$.
The algorithm uses a stack, $Agenda$, to memoize all new pairs $(q_1,s_1)\in\states^{A\times S}$ that need to be processed, and a hash set $Store$  that saves all processed states $\states^{A \times S}$. We aim not to push the same pair more than once in the $Agenda$. For this, membership to the $Store$ is checked before an element is pushed to the $Agenda$. For each pair
popped from the stack $Agenda$, the algorithm does the following:
for each letter $a\in\Sigma$
it computes the sets $Q= \{q_2 \mid q_1 \xrightarrow a q_2 \in \Delta^A\}$
and $R= \{s_2 \mid s_1\xrightarrow a s_2\in\Delta^S\}$
and then adds the subset of states of $Q\times R$ that were not
stored in the hash set $Store$ to the agenda. Since $A$ and $S$ are deterministic, there is
at most one such pair, so the time for treating one pair on the
agenda is in expected amortized time $O(|\Sigma|)$.
The overall number of elements in the agenda
will be $|\states^{A\times S}|$.  Note that $Q$ and $R$ can be
computed in $O(1)$ after preprocessing $A$ and $S$ in time $O(|A|+|S|)$. Therefore, we will have a total time of the algorithm in $O(|\states^{A\times S}| |\Sigma|  + |A| + |S| )$.
\end{proof}
\begin{corollary}\label{scl-words-cor}
  For any two \DFAs $A$ and $S$  with alphabet $\Sigma$  schema-based cleaning $\sclS(A)$
  can be computed in expected amortized time $O(|\states^{A\times S}| |\Sigma|  + |A| + |S| )$.
\end{corollary}

\begin{proof}
  By Definition \ref{sclpro-words} it is sufficient to compute the projection of the accessible
  product $A \times S$. By  Proposition \ref{compl-scl-words}
  the product can be computed in time $O(|\states^{A\times S}| |\Sigma|
  + |A| + |S| )$. Its size cannot be larger than its computation
  time. The projection
  can be computed in linear time in the size of $A\times S$, so
  the overall time is in $O(|\states^{A\times S}| |\Sigma|
  + |A| + |S| )$ too.
\end{proof}

\begin{figure}[t]
\begin{sml}
fun $A \times S$ =
  let $Store=hashset.new(\emptyset)$ and $Agenda=list.new()$ and $Rules=hashset.new(\emptyset)$
  if $\init^A=\{q_0\}$ and  $\init^S=\{s_0\}$ then $Agenda.add((q_0,s_0))$  
  while $Agenda.notEmpty()$ do 
    let $(q_1,s_1)= Agenda.pop()$
    for $a \in \Sigma$ do
      let $Q=\{ q_2 \mid \SAir{q_1}{a}{q_2} \in \Delta^A\}$    $R=\{ s_2 \mid \SAir{s_1}{a}{s_2} \in \Delta^S\}$
      for $q_2 \in Q$ and $s_2 \in R$ do
    	  $Rules.add(\SAir{(q_1,s_1)}{a}{(q_2,s_2)})$
    	  if not $Store.member((q_2,s_2))$
    	    then $Store.add((q_2,s_2))$ $Agenda.push((q_2,s_2))$
  let $\init^{A\times S}=\{(q_0,s_0) \mid (q_0,s_0) \in Store \}$ and $F^{A \times S}=\{(q,s) \mid (q,s) \in Store, q\in F^A, s \in F^S\}$
  return $(\Sigma,Store.toSet(),Rules.toSet(),\init^{A\times S},F^{A \times S})$
\end{sml}
\caption{\label{algo-int-rules-words} An algorithm computing the accessible product of \DFAs $A$ and $S$.}
\end{figure}

\section{Schema-Based Determinization}
\label{sec:schema-det-words}

Schema-based cleaning after determinization becomes impossible
in practice if the automaton obtained by determinization is too
big. We therefore show next how to integrate schema-based cleaning
into automata determinization directly.

The schema-based determinization of $A$ with respect to schema $S$
extends on accessible determinization $\det(A)$. The idea is to run
the schema $S$ in parallel with $\det(A)$, in order to keep only those state 
$Q\in\states^{\det(A)}$  that can be aligned to some
state $s\in\states^S$. In this case we write
$Q\sim s$.

\begin{figure}[t]
$$
\infer[]
{
Q \in \Qinith^{\detS(A)}  \quad Q\sim s
}
{
 Q \in \Qinith^{det(A)} & \Qinith^S=\{s\}
}
\qquad 
\infer[]
{
Q \in \states^{\detS(A)}
}
{
Q\sim s
}
\qquad
\infer[]
         {Q\in \FIN^{\detS(A)}}
         {Q\in \FIN^{\det(A)}\qquad  s\in \FIN^S \qquad   Q\sim s } 
$$
$$
\infer[]
   {
Q\xrightarrow{a} Q' \in\Delta^{\detS(A)} \qquad Q'\sim s'
    }
  {
\begin{array}{c}
 \SAir{Q}{a}{Q'} \in \Delta^{det(A)} \qquad Q\sim s \qquad \SAir{s}{a}{s'} \in \rules^S 
\end{array}
    }
$$
\caption{\label{scd-words}Schema-based determ. $  \detS(A)=
  (\Sigma,\states^{\detS(A)}, \Delta^{\detS(A)}, \Qinith^{\detS(A)}, \Qfin^{\detS(A)})
$.}
\end{figure}

The schema-determinization $\detS(A)$ is defined
in \Figure{scd-words}. The automaton  $\detS(A)$ permits to go from 
any subset $Q\in\states^{\det(A)}$ and letter $a\in \Sigma$ to the set of
states $Q'=a^{\Delta^{\det(A)}}(Q)$, under the condition
that there exists schema states $s,s'\in\states^S$ such that 
$Q\sim s$ and $s\xrightarrow{a} s'$. In this case $Q'\sim s'$
is inferred.

\newcommand\QN[1]{\text{QN#1}\xspace}

\newcounter{theoremproof}
\newtheoremrep{keytheoremproof}[theoremproof]{Claim}
\newcommand\reftheoremproof[1]{\ref{thproof-#1}\xspace}

\newcounter{detS}
\newtheoremrep{theo}{Theorem}
\newcommand\refinsS[1]{\ref{detS#1}\xspace}

\begin{theo}[Correctness]\label{detS1} 
$\detS(A) = \schemaclean{\det(A)}{S}$ for any \NFA $A$ and \DFA $S$
with the same alphabet. 
\end{theo}

The theorem states that schema-based determinization yields the
same result as accessible determinization followed by schema-based cleaning.

\label{sec:schema-clean-proof-words}
For the correctness proof 
we collapse the two systems of inference rules for 
accessible products and projection 
into a single rule system. This
yields the rule systems for schema-based
cleaning in \Figure{prf2-words}. 
%
%
\begin{figure}[t]
$$
\begin{array}{c}
  \infer[]
  {
  	q \in \Qinith^{\SCL} \quad 	(q,s) \in \states^{A\hattimes S}
  }
  {
	q \in \Qinith^A \quad s \in   \Qinith^S
  }
  \qquad
   \infer[]
  {
  q \in  \Qfin^{\SCL}
  }
  {
  q \in \Qfin^A & s \in \Qfin^S & (q,s) \in \states^{A\hattimes S}
  }
  \\[.2em]
  \infer
  {
  	q \in \states^{\SCL}
  }
  {
  (q,s) \in \states^{A\hattimes S}
  }
  \qquad
  \infer[]{
   \SAir{q_1}{a}{q_2} \in \Delta^{\SCL} \quad (q_2,s_2) \in \states^{A\hattimes S}
  }
  {
  \SAir{q_1}{a}{q_2} \in \Delta^{A} \quad \SAir{s_1}{a}{s_2} \in \Delta^{S} \quad (q_1,s_1) \in \states^{A\hattimes S}
  }
\end{array}
$$
$$
\SCL=(\Sigma,\states^{\SCL},\Delta^{\SCL},\Qinit^{\SCL},\Qfin^{\SCL})
$$
\caption{\label{prf2-words}A collapsed rule systems for schema-based cleaning $\SCL$. }
\end{figure}
\begin{figure}[t]
  \begin{sml}
fun detS($A$,$S$) =
  let $Store=hashset.new(\emptyset)$ and $Agenda=list.new()$ and $Rules=hashset.new(\emptyset)$
  if $\init^A\not=\emptyset$ and  $\init^S=\{s_0\}$ then $Agenda.add(\init^A\sim s_0)$   
  while $Agenda.notEmpty()$ do 
    let $(Q_1\sim s_1) = Agenda.pop()$
    for $a \in \Sigma$ do
      let $P=\{ Q_2 \mid \SAir{Q_1}{a}{Q_2} \in \Delta^{\det(A)}\}$ and $R=\{ s_2 \mid \SAir{s_1}{a}{s_2} \in \Delta^S\}$
      for $Q_2 \in P$ and $s_2 \in R$ do $Rules.add(\SAir{Q_1}{a}{Q_2})$
         if not $Store.member(Q_2\sim s_2)$ 
          then $Store.add(Q_2\sim s_2)$    $Agenda.push(Q_2\sim s_2)$
  let $\init^{\detS(A)}=\{Q \mid Q \sim s \in Store, Q \cap \init^A\not=\emptyset\}$ and $F^{\detS(A)}=\{Q \mid Q \sim s \in Store, Q \cap F^A \not= \emptyset\}$
  return $(\Sigma,Store.toSet(),Rules.toSet(),\init^{\detS(A)},F^{\detS(A)})$
\end{sml}
  \caption{\label{algo-schema-det-words} An algorithm for
    schema-based determinization $\detS(A)$  of an \NFA $A$ and a \DFA schema $S$.}
 \end{figure}
%
The rules there define the automaton $\SCL$, that we annotate 
with a hat, in order to distinguish it from the previous
automaton $\scl{S}(A)$. The rules also infer 
judgements $(q,s)\in\states^{A \hattimes S}$ that we distinguish
by a hat from the previous judgments 
$(q,s)\in\states^{A\times S}$ of the accessible product.
The next proposition shows that the system of collapsed
inference rules indeed redefines the schema-based cleaning.

\begin{propositionrep}\label{collapse-words}
For any two \NFAs $A$ and $S$ with the same alphabet:
$$\scl{S}(A)=\SCL
\quad\text{\ and\ }\quad  \states^{A\times S}=\states^{A\hattimes S}$$
\end{propositionrep}

  \begin{toappendix}
\begin{proof}
The two equations are shown by the following four lemmas. The
judgements with a hat there are to be inferred by the collapsed 
system of inference rules in \Figure{prf2-words}, while the 
other judgments are to be inferred with the 
rule system for accessible products in \Figure{fig:product_sha_words}.

\end{proof}

\begin{lemmarep}\label{collapselemma0}
$q \in \Qinith^{\SCL}$ iff $q \in \Qinith^{\scL}$.
\end{lemmarep}

\begin{proof}
The rule systems of accessible product, projection, and the collapsed
system can be used as following :
$$
  \infer[]
  {
  	q \in \Qinith^{\SCL} 
  }
  {
	q \in \Qinith^A \quad s \in   \Qinith^S
  }
  \quad
  \Infer
  {
    q \in \Qinith^A & s \in \Qinith^S
  }
  {
    \Infer{(q,s) \in \Qinith^{A\times S}}{q \in \Qinith^{ \scL} }
  }
  $$
\end{proof}

\begin{lemmarep}\label{collapselemma1}
$(q,s) \in \states^{A\hattimes S}$ iff $(q,s) \in \states^{A \times S}$.
\end{lemmarep}

\begin{proof}
We proof for all $n\ge 0$ that if $(q,s) \in \states^{A\hattimes S}$ 
has a proof tree of size $n$ then there exists a proof tree
for $(q,s) \in \states^{A \times S}$. The proof is by induction
on $n$.

In the case of the rules of the initial states, $(q,s) \in \states^{A\hattimes S}$ is inferred directly whenever  $(q,s) \in \states^{A \times S}$ and vice versa, using the following:
$$
 \infer[]
  {
  (q,s) \in \Qinith^{A\times S}
\qquad
  (q,s) \in \states^{A\times S}
  }
  {
  q \in \Qinith^A & s \in \Qinith^S
  }
  \qquad
 \infer[]
  {
  	q \in \Qinith^{\SCL} \quad 	(q,s) \in \states^{A\hattimes S}
  }
  {
	q \in \Qinith^A \quad s \in   \Qinith^S
  } 
$$

If $(q,s) \in \states^{A\hattimes S}$ is inferred by the
internal rule of the collapsed rule system in \Figure{prf2-words}. 
Then the proof tree has the following form
for some proof tree $T_1$: 
$$
 \Infer
 {
  \SAir{q_1}{a}{q_2} \in \Delta^{A} \quad \SAir{s_1}{a}{s_2} \in \Delta^{S} \quad 
  \Infer{
  		T_1
 	 }
  	{
  		(q_1,s_1) \in \states^{A\hattimes S}
 	 }
  }
  {
  (q_2,s_2) \in \states^{A\hattimes S}
  }
$$
This shows that there is a smaller proof tree $T_1$ for inferring $(q_1,s_1) \in
\states^{A\hattimes S}$. So by induction hypothesis applied to $T_1$, there 
exists a proof tree $T'_1$ for inferring $(q_1,s_1) \in
\states^{A\times S}$ with the proof system of accessible
products in \Figure{fig:product_sha_words}:
$$
 \Infer{
  		T'_1
 	 }
  	{
  		(q_1,s_1) \in \states^{A\times S}
 	 }
$$
Therefore, we also have the following proof tree for $ (q_2,s_2) \in
\states^{A\times S}$ with the internal rule for the accessible product:
$$
 \Infer
 {
  \SAir{q_1}{a}{q_2} \in \Delta^{A} \quad \SAir{s_1}{a}{s_2} \in
  \Delta^{S} \quad   	 \Infer{
  		T'_1
 	 }
  	{
  		(q_1,s_1) \in \states^{A\times S}
 	 }
  }
  {  
    (q_2,s_2) \in \states^{A\times S}  
  }
$$

For the inverse direction, if $(q,s) \in \states^{A\times S}$ is inferred by the
internal rule of the accessible product rule system in \Figure{fig:product_sha_words}. 
Then the proof tree has the following form
for some proof tree $T_1$: 
$$
 \Infer
 {
  \SAir{q_1}{a}{q_2} \in \Delta^{A} \quad \SAir{s_1}{a}{s_2} \in \Delta^{S} \quad 
  \Infer{
  		T_1
 	 }
  	{
  		(q_1,s_1) \in \states^{A\times S}
 	 }
  }
  {
  (q_2,s_2) \in \states^{A\times S}
  }
$$
This means that there is a smaller proof tree $T_1$ for inferring $(q_1,s_1) \in
\states^{A\times S}$. By induction hypothesis applied to $T_1$, there 
exists a proof tree $T'_1$ for inferring $(q_1,s_1) \in
\states^{A\hattimes S}$ with the collapsed system in \Figure{prf2-words}:
$$
 \Infer{
  		T'_1
 	 }
  	{
  		(q_1,s_1) \in \states^{A\hattimes S}
 	 }
$$
which leads to the following proof tree for $ (q_2,s_2) \in
\states^{A\times S}$ with the internal rule for the collapsed system:
$$
 \Infer
 {
  \SAir{q_1}{a}{q_2} \in \Delta^{A} \quad \SAir{s_1}{a}{s_2} \in
  \Delta^{S} \quad   	 \Infer{
  		T'_1
 	 }
  	{
  		(q_1,s_1) \in \states^{A\hattimes S}
 	 }
  }
  {  
    (q_2,s_2) \in \states^{A\hattimes S}  
  }
$$
 
\end{proof}

\begin{lemmarep}\label{collapselemma2}
$\SAir{q_1}{a}{q_2} \in \Delta^{\SCL}$ iff $\SAir{q_1}{a}{q_2} \in \Delta^{ \scL}$.
\end{lemmarep}

\begin{proof}
We prove for all $n\ge 0$ that, if $\SAir{q_1}{a}{q_2} \in \Delta^{\SCL}$ 
has a proof tree of size $n$, then there exists a proof tree
for $\SAir{q_1}{a}{q_2} \in \Delta^{ \scl{S}{A}}$ and vice versa. The proof is by induction
on $n$.
\\
\\
If $\SAir{q_1}{a}{q_2} \in \Delta^{ \SCL}$  is inferred by the internal rule of the collapsed system, the proof tree will have the following for some tree $T_1$:
$$
 \Infer
 {
  \SAir{q_1}{a}{q_2} \in \Delta^{A} \quad \SAir{s_1}{a}{s_2} \in \Delta^{S} 
  \quad 
  \Infer{T_1}{(q_1,s_1) \in \states^{A \hattimes S}}
  }
  {
  \SAir{q_1}{a}{q_2} \in \Delta^{ \SCL}
  }
$$
By Lemma \ref{collapselemma1} and the rule of internal rules of the accessible product rule system:
$$
 \Infer
 {
  \SAir{q_1}{a}{q_2} \in \Delta^{A} \quad \SAir{s_1}{a}{s_2} \in \Delta^{S} 
  \quad 
  \Infer{T'_1}{(q_1,s_1) \in \states^{A \times S}}
  }
  {
  \SAir{(q_1,s_1)}{a}{(q_2,s_2)} \in \Delta^{A\times S}
  }
$$

For the inverse direction, if $\SAir{q_1}{a}{q_2} \in \Delta^{\scL}$  is inferred by the internal rule of the accessible product, the proof tree will have the following for some tree $T_1$:
$$
 \Infer
 {
  \SAir{q_1}{a}{q_2} \in \Delta^{A} \quad \SAir{s_1}{a}{s_2} \in \Delta^{S} 
  \quad 
  \Infer{T_1}{(q_1,s_1) \in \states^{A \times S}}
  }
  {
   \Infer{\SAir{(q_1,s_1)}{a}{(q_2,s_2)} \in \Delta^{A\times S}}{\SAir{q_1}{a}{q_2} \in \Delta^{\scL} }
  }
$$
By lemma \ref{collapselemma1} and the rule of internal rules of the collapsed system:
$$
 \Infer
 {
  \SAir{q_1}{a}{q_2} \in \Delta^{A} \quad \SAir{s_1}{a}{s_2} \in \Delta^{S} 
  \quad 
  \Infer{T'_1}{(q_1,s_1) \in \states^{A \times S}}
  }
  {
\SAir{q_1}{a}{q_2} \in \Delta^{\SCL}
  }
$$

\end{proof}

\begin{lemmarep}\label{collapselemma4}
$q \in \states^{\SCL}$ iff $q \in \states^{\scL}$ and
$q \in \Qfin^{\SCL}$ iff $q \in \Qfin^{ \scL}$.
\end{lemmarep}
\begin{proof}
We start proving $q \in \states^{\SCL}$ iff $q \in \states^{\scL}$.
By Lemma \ref{collapselemma1}, and rules of construction of the accessible product, projection, and collapsed systems, this lemma holds for some proof trees $T$ and $T'$ as follows:
$$
 \Infer
  {
   T
  }
  {
 	\Infer
 	{
 	  (q,s) \in \states^{A\times S}
 	}
 	{
 	q \in \states^{\scL}
 	}
  }
  \quad
  \Infer
  {
   T'
  }
  {
  	q \in \states^{\SCL}
  }
$$

Finally, we show $q \in \Qfin^{\SCL}$ iff $q \in \Qfin^{ \scL}$.
Using Lemma \ref{collapselemma1}, there exists some proof trees $T$ and $T'$ that infers $(q,s) \in \states^{A\times S}$ and $(q,s) \in \states^{A\hattimes S}$ in both ways and therefore having the following form of rules:
$$
 \Infer
  {
    q \in \Qfin^A & s \in \Qfin^S & 
    \Infer{T}{(q,s) \in \states^{A\hattimes S}}
  }
  {
  q \in  \Qfin^{\SCL}
  }
  \quad
  \Infer
  {
   q \in \Qfin^A & s \in \Qfin^S &  
   \Infer{T'}{(q,s) \in \states^{A\times S}}
  }
  {
  	\Infer{(q,s) \in \Qfin^{A\times S}}{q \in  \Qfin^{\scL}}
  }
$$

\end{proof}
\end{toappendix}


\label{sec:schema-det-proof}

\medskip

\noindent\textit{Proof of Correctness Theorem \ref{detS1}.} 
Instantiating the system of collapsed rules for schema-based cleaning 
from \Figure{prf2-words} with $\det(A)$ for $A$ yields the rule
system in \Figure{inst-words}.
%
\begin{figure}[t]
$$
\begin{array}{c}
\infer[]
  {
  	Q \in \Qinith^{\hatscl{S}(det(A))} \quad (Q,s) \in \states^{det(A) \hattimes S}
  }
  {
	Q \in \Qinith^{det(A)} \quad s \in   \Qinit^S
  }
  \quad
   \infer[]
  {
  Q \in  \Qfin^{\hatscl{S}(det(A))}
  }
  {
  Q \in \Qfin^{det(A)} & s \in \Qfin^S & (Q,s) \in \states^{det(A)\times S}
  }
  \\[.2em]
   \infer
  {
	  Q \in \states^{\hatscl{S}(det(A))}  
  }
  {
  	(Q,s) \in \states^{det(A) \hattimes S}
  }
\quad  \infer[]{
   \SAir{Q_1}{a}{Q_2} \in \Delta^{\hatscl{S}(det(A))} \quad (Q_2,s_2) \in \states^{det(A)\hattimes S}
  }
  {
  \SAir{Q_1}{a}{Q_2} \in \Delta^{det(A)} \quad \SAir{s_1}{a}{s_2} \in \Delta^{S} \quad (Q_1,s_1) \in \states^{det(A)\times S}
  }
\end{array}
$$
$$
\DSCL=(\Sigma,\states^{\DSCL},\Delta^{\DSCL},\Qinit^{\DSCL},\Qfin^{\DSCL})
$$
\caption{\label{inst-words}Instantiation of the collapsed rules
  for schema-based cleaning from \Figure{prf2-words} with $\det(A)$.}
\end{figure}
  %
We can identify the instantiated collapsed system for $\hatscl{S}(det(A))$ 
with that for $\detS(A)$ in \Figure{scd-words}, by  identifying the
judgements $ (Q,s) \in \states^{det(A)\hattimes S}$ with
judgments $Q\sim s$. After renaming the predicates, 
the inference rules for the corresponding judgments 
are the same. Hence $\hatscl{S}(det(A))=\detS(A)$,
so that Proposition \ref{collapse-words} implies
$\scl{S}(det(A))=\detS(A)$. \qed


\begin{proposition}\label{compl-detS-words}
  The schema-based determinization $\detS(A)$ for a \NFA $A$ and
  a \DFA $S$ over $\Sigma$ can be computed in expected amortized time
  $O(|\states^{\det(A) \times S}||\Sigma| + |\states^{\detS(A)}||\Delta^A| + |A| +|S|)$.
\end{proposition}

\begin{proof}
     An algorithm computing the
      fixed points of the inference rules of schema-based determinization
      from \Figure{scd-words} is given in \Figure{algo-schema-det-words}.
      It refines the algorithm computing the accessible product with
      on-the-fly determinization and projection.
       
  On the stack $Agenda$, the algorithm stores alignments $Q\sim s$
  such that $(Q,s)\in \states^{\det(A)\times S}$ that were not
  considered before. Transition rules of $\detS(A)$ are collected in
  hash set $\mathit{Rules}$, using the dynamic perfect hashing
  aforementioned.
  The alignments $Q_1 \sim s_1$ popped from the agenda are processed as follows:
  For any letter $a\in\Sigma$, the sets
  $R=\{Q_2\mid Q_1\xrightarrow a Q_2\in \Delta^{\det(A)}\}$
  and $P=\{s_2\mid s_1\xrightarrow a s_2\in \Delta^S\}$ are
  computed.
 One then pushes all new pairs $Q_2\sim s_2$ with $Q_2\in P$
  and $s_2\in R$ into the agenda, and adds $ Q_1\xrightarrow a Q_2$ to the set $\mathit{Rules}$.
  Since $S$ and $\det(A)$ are deterministic
  there is at most one pair $(Q,s) \in P\times R$ for $Q_1$ and $s_1$.
  So the time for treating one pair on the
  agenda is in $O(|\Sigma|)$ plus the time
  for building the needed transition rules
  of $\det(A)$ from $\Delta^A$ on the fly.
  The time for the on the fly computation
  of transition rules of $\det(A)$ is in
  time $O(|\states^{\detS(A)}||\Delta^A|)$.
  The overall number of pairs on the agenda
  is at most $|\states^{\det(A)\times S}|$ so the main
  while loop of the algorithm
  requires time in $O(|\states^{\det(A)\times S}||\Sigma|)$
  apart from on the fly determinization. 
    \end{proof}

  By Proposition \ref{comp-detacc-words}, computing $\det(A)$ requires
  time $O(|\states^{det(A)}|\MULTC |\Delta^A| + |A|)$. Therefore, with Proposition \ref{compl-scl-words},
    the accessible product $\det(A)\times S$ can be computed from $A$ and $S$ in
    time $O(|\states^{\det(A)\times S}| |\Sigma| +
    |\states^{det(A)}|\MULTC |\Delta^A|+|A|+|S|)$.
    Since 
    $\states^{\detS(A)}\subseteq \states^{\det(A)}$
    the proposition shows that schema-based determinization is at most as efficient in the worst case 
  as accessible determinization followed by schema-based cleaning.
If $|\states^{\det(A)\times S}||\Sigma|<|\states^{\det(A)}||\Delta^A| $ then it is more efficient, since schema-based determinization
  avoids the computation of $\det(A)$
  all over. Instead, it only computes the accessible product $\det(A)\times S$, which may be way smaller,
  since exponentially many states of $\det(A)$ may not be aligned to any state of $S$. Sometimes, however, the
  accessible product may be bigger. In this case, schema-based determinization may be more costly than
  pure accessible determinization, not followed by schema-based cleaning.


\section{Stepwise Hedge Automata for Nested Words}
\label{sec:nw}

We next recall \SHAs \cite{Sakho} for defining
languages of nested words, regular schemas and queries.
Nested words generalize on words by adding parenthesis that must be
well-nested. While containing words natively, they also generalize
on unranked trees, and hedges. We restrict ourselves to nested
words with a single pair of opening and closing parenthesis
$\OP$ and $\CL$. 
Nested words over a finite alphabet $\Sigma$ of internal letters
have the following abstract syntax. 

\noindent
\begin{minipage}{\textwidth}
\begin{minipage}{.83\textwidth}
$$
\begin{array}{rcl}
  \nw,\nw' \in \NW&::=& \epsilon \mid a \mid
                             \TC\nw \mid
                             \nw\SEQ \nw'\qquad\text{
                             where }a\in\Sigma
\end{array}
$$
  We assume that concatenation $\cdot$ is associative and that
  the empty word $\epsilon$ is a neutral element,
  that is $w\cdot(w'\cdot w'')=(w\cdot w')\cdot w''$
  and $\epsilon\SEQ \nw=\nw=\nw\SEQ\epsilon$.
Nested words can be identified with hedges, i.e., words of unranked trees and
letters from $\Sigma$. Seen as a graph, the inner nodes are labeled by the tree constructor $\OP\CL$ and the leafs by
symbols in $\Sigma$ or the tree constructor.  For instance $\TC{ a \SEQ \TC{b} \SEQ \epsilon} \SEQ c \SEQ \TC{d
    \SEQ  \TC{\epsilon}}$ corresponds to the hed-
\end{minipage}
\begin{minipage}{.16\textwidth}
\footnotesize{\begin{forest}
for tree={s sep=(3-level)*0.6mm}
        [$\TC{}$
             [a]
             [$\TC{}$
               [b]]]
\end{forest}
\begin{forest}
        [c]
\end{forest}
\begin{forest}
        [$\TC{}$
             [d]
             [$\TC{}$]]
\end{forest}}
\end{minipage}
\end{minipage}

\noindent ge~on~the~right. A nested word of type \emph{tree}
has the form $\TC{h}$.
Note that dangling parentheses are ruled out and
that labeled parentheses can be simulated by using internal
letters. \ignore{Binary trees with a single 
constant $\epsilon$ and a single binary function symbol $\OP .,.\CL$
are nested words of the form 
$
\begin{array}{rcl}
  t,t'  \in\mathcal{T} &::=& \OP \epsilon\CL  \mid  \OP  t \SEQ t' \CL 
\end{array}
$.
Words in $\Sigma^*$ are captured by nested words of the form
$
\begin{array}{rcl}
  w,w' \in \Sigma^*&::=& \epsilon \mid  a \mid w \SEQ w' 
\end{array}
$ 
where $a$ in $\Sigma$. Nested words also generalize over
sequences of binary trees, which have the form:
$
\begin{array}{rcl}
  l,l' \in \mathcal{S}&::=& \epsilon\mid  t \mid l \SEQ l' 
\end{array}
$ 
where $t\in\mathcal{T}$. Labeled unranked trees such as 
$a(b(),c())$ can be represented by the nested word 
$\OP a \SEQ \OP b \CL \SEQ \OP c \CL\CL$. In this way, the labeled tree $a()$
is represented by the nested word $\OP a \CL$.}
\XML documents are labeled unranked trees, for instance:
$
\OP a\  name=``uff"\CL \OP b\CL isgaga \OP d/\CL \OP /b\CL \OP c/\CL \OP /a\CL
$. Labeled unranked trees satisfying the \XML data model can be represented as
nested words over an alphabet that contains the \XML node-types $(elem, attr,
text,$ $\ldots)$, the \XML names of the document $(a,\ldots,d,name)$, and the
characters of the data values, say UTF8. For the above example, we get
the nested word
$
  \OP elem\SEQ a\SEQ  \OP attr\SEQ name\SEQ u\SEQ f\SEQ f\CL \OP elem\SEQ b\SEQ \OP
  text\SEQ  i\SEQ  s\SEQ  g\SEQ  a\SEQ  g\SEQ  a\CL \OP elem\SEQ  d\CL \CL \OP elem\SEQ c\CL \CL
$

\label{sec:stepwise}

\begin{definition}
A \SHA  is a tuple $\SAmath=(\Sigma,\states,\Delta,\Qinith,\Qfin)$
where $\Delta=(\Delta',@^{\Delta},\TIR^{\Delta})$ such that
$(\Sigma,\states,\Delta',\Qinith,\Qfin)$ is a \NFA,
$\TIR^\Delta \subseteq \states$ is a set of
tree initial states and
$\SAtrans{@} \subseteq \states^3$
a set of apply rules.
\end{definition}

\SHAs can be drawn as graphs while extending on the graphs of \NFAs.   A tree initial state $q\in \Qinitt$ is drawn as a node  $\ItstatePic{\ q\ }$ with an incoming
tree arrow. An applyrule
$(q_1,q,q_2)\in \SAtrans{@}$ is drawn as a blue edge  $\statePic{q_1}{\color{blue}\xrightarrow{q}}\statePic{q_2}$ that is labeled
by a state $q\in\states$ rather than a letter $a\in \Sigma$. It
states that a nested word in state $q_1$ can be extended by a tree in
state $q$ and become a nested word in state $q_2$. 

For instance, the \SHA \onexnw is drawn graphically 
in \Figure{one-x-step}. It accepts all nested
words over $\Sigma\uplus\{x,\neg x\}$ that contain
exactly one occurrence of letter $x$. Compared
to the \NFA \onexneg from \Figure{DFAONEX}, 
the \SHA \onexnw contains three additional apply rules
$(0,0,0)$, $(0,1,1)$, $(1,0,1) \in @^{\Delta^\onexnw} $
%
for reading the states assigned to
subtrees. The state $0$
is chosen as the single tree initial state.

Transitions for \NFAs on words can be lifted to transitions for
\SHAs of the form
$q\xrightarrow{\nw} q' \wrt \Delta$ where $\nw\in\NW$ and $q,
q'\in\states$.
For this, we add the following inference rule to the previous
rules for \NFAs:
$$
\infer{q_1\xrightarrow{\langle\nw\rangle  } q_2\wrt \Delta}
{
q'\in \Qinittgen{\Delta}
\qquad
q' \xrightarrow{w} q \wrt\Delta
\qquad
(q_1,q,q_2) \in @^\Delta}
$$
The rule says that a tree $\langle\nw\rangle$ can transit from a state
$q_1$ to a state $q_2$ if
there is an apply rule
$(q_1,q,q_2)\in @^\Delta$ so that $\nw$ can transit from
some tree initial state $q'\in\Qinitt$ to $q$. Otherwise, the language $\L(A)$
of nested words 
accepted by a \SHA $A$ is defined as in the case of \NFAs.

\begin{definition}\label{det}
A \SHA $(\Sigma,\states,\Delta,\Qinith,\Qfin)$ is \emph{deterministic} or
equivalently a \dSHA if it satisfies:
\begin{itemize}
\item $\Qinith$ and $\Qinitt$ both contain at most one element,
  \item $\SAtrans{a}$ is a partial function
    from $\states$ to $\states$ for all $a\in\Sigma$, and
\item $\SAtrans{@}$ is a partial function
    from $\states\times \states$  to $\states$.
\end{itemize}
\end{definition}
Note that if $A$ is a \dSHA and $\Delta=(\Delta',@^\Delta,\Qinitt)$ then $A'=(\Sigma,\states,\Delta',\Qinith,\Qfin)$ is a \DFA. Conversely any \DFA $A'$
defines a \dSHA with $@^\Delta=\emptyset$ and $\Qinit=\emptyset$.
For instance,
the \SHA $\onexnw$ in \Figure{one-x-step} contains the \DFA \onexneg from \Figure{DFAONEX} with $\Sigma$ instantiated by $\Sigma\uplus\{x\}$.

A schema for nested words over $\Sigma$ is a \dSHA over
$\Sigma$. Note that schemas for nested words generalize over schemas
of words, since \dSHAs generalize on \DFAs.
%
%
\begin{figure}[t]
$$
  \infer{
  \Qinittgen{\Delta^A}\in \accdet\A
  }
  {
   \Qinittgen{\Delta^A}\not=\emptyset
  }
  \qquad 
 \infer{
Q_1@Q_2 \to  Q'\in\Delta^{det(A)} \qquad Q'\in \accdet\A
    }
  {
\begin{array}{c}
   Q_1 \in \accdet\A\qquad   Q_2 \in \accdet\A\\
   Q' =\{ q'\in\states^A\mid   q_1@q_2\to q' \in \Delta^A, \ q_1 \in Q_1, q_2 \in Q_2\}  \not=\emptyset 
   \end{array}
 }
$$
\caption{\label{extaccdet} Accessible
  determinization $\det(A)$ lifted from \NFAs to \SHAs.}
  \end{figure}
\begin{toappendix}
  \begin{figure}[t]
  \begin{sml}
fun detSHA(A) =
  let $Store=hashset.new(\emptyset)$
  let $Agenda=list.new()$ and $Rules=hashset.new(\emptyset)$
  if $\init^A\not=\emptyset$ then $Agenda.add(\init^A)$   
  while $Agenda.notEmpty()$ do 
    let $(Q) = Agenda.pop()$
    let $h$ be an empty hash table with keys from $\Sigma$.
      // the values will be nonempty hash subsets of $\states^A$
    for $q\xrightarrow{a} q' \in \Delta^A$ such that $q\in Q$ do
      if $h.get(a)=undef$ then $h.add(a,hashset.new(\emptyset))$
    $(h.get(a)).add(q')$ 
    for $(a,Q')$ in $h.tolist()$ do $Rules.add(\SAir{Q}{a}{Q'})$
      if not $Store.member(Q')$  then $Store.add(Q')$ $Agenda.push(Q')$
    for $Q_1 \in Store$ do
      let $Q'=\{q' \mid \SAar{q}{q_1}{q'}, q_1 \in Q_1, q \in Q\}$ 
      if $Q' \not= \emptyset$ then $Rules.add(\SAar{Q}{Q_1}{Q'})$
    	  if not $Store.member(Q')$ then $Store.add(Q')$ $Agenda.push(Q')$ 
      let $Q''=\{q'' \mid \SAar{q_1}{q}{q''}, q_1 \in Q_1, q \in Q\}$ 
      if $Q'' \not= \emptyset$ then $Rules.add(\SAar{Q_1}{Q}{Q''})$
  let $\init^{det(A)}=\{Q \mid Q \in Store, Q \cap \init^A \not= \emptyset\}$ and $F^{det(A)}=\{Q \mid Q \in Store, Q \cap \Qfin^A \not= \emptyset\}$
  return $(\Sigma,Store.toSet(),Rules.toSet(),\init^{det(A)},F^{det(A)})$
\end{sml}

  \caption{\label{algo-accdet-nwords} An algorithm for accessible determinization of \SHAs.}
\end{figure}
\end{toappendix}
The rules for the accessible determinization $\det(A)$ of
a \SHA $A$ in \Figure{extaccdet} extend on those
for \NFAs in \Figure{accdet}. As for words, $\det(A)$ is
always determinstic, recognizes the same language
as $A$, and contains only accessible states. The complexity of accessible determinization in case of \SHA go similarly to \DFA, however, the apply rules will introduce quadratic factor in the number of states. 

\begin{propositionrep}\label{comp-detacc-nwords}
The accessible determinization of a \SHA can be computed in expected amortized time $O(|\states^{\det(A)}|^2\MULTC|\Delta^A| + |A|)$.
\end{propositionrep}

\begin{appendixproof}
  An algorithm
      for computing the fixed points of the inference rules of
      accessible determinization of a \SHA is presented in
      \Figure{algo-accdet-nwords}. It extends on the case of \NFAs
      with the same data structures. It uses dynamic perfect hashing
      for the hash sets. The additional treatment of
      apply rules, that dominates the complexity of the algorithm, works as follows: for each $Q
      \in \states^{det(A)}$ in the $Agenda$ and each state $Q1 \in
      \states^{det(A)}$ in the $Store$, it computes the sets  $Q'=\{q'
      \mid \SAar{q}{q_1}{q'}, q_1 \in Q_1, q \in Q\}$ and $Q''=\{q''
      \mid \SAar{q_1}{q}{q''}, q_1 \in Q_1, q \in Q\}$ and puts all
      new non-empty sets in both the $Agenda$ and the $Store$, while
      adding dynamically the generated apply rules in the hash set
      $Rules$. Again, the overall number of elements in the agenda
      will be $|\states^{det(A)}|$, requiring time
      in $O(|\states^{det(A)}|^2\MULTC|\Delta^A|)$. With a precomputation time of $A$ in $O(|A|)$, the total computation will be in $O(|\states^{\det(A)}|^2\MULTC|\Delta^A| + |A|)$.   
\end{appendixproof}

The notions of monadic query functions $\QUERY$ can be lifted from
words to nested words, so that it selects nodes of the graph of
a nested word. For this, we have to fix one of manner possible
manners to define identifiers for these nodes. The set of nodes of a nested word $w$
is denoted by $\nodes(w)\subseteq\Nat$.

\begin{toappendix}
\begin{figure}[tb]
\centering
  \includegraphics[scale=0.25]{dSHA-QN7.png}
  \caption{\label{dSHA-QN7}A $\onexnw$-cleaned minimal \dSHA for the \XPath
    query $\QNN7$.}
\end{figure}
\end{toappendix}
For indicating the selection of node $\pi\in\nodes(w)$, we insert the
variable $x$ into the sequence of letters following the opening parenthesis
of $\pi$. If we don't want to select $\pi$, we insert
the letter $\neg x$ instead. For any nested word $w$
with alphabet $\Sigma$, the nested word $w[\pi/x]$ obtained
by insertion of $x$ or $\neg x$ at a node $\pi\in\nodes(w)$ has alphabet
$\Sigma\uplus\{x,\neg x\}$. As before, we define
$L_\QUERY=\{w*[\pi/x]\mid w\in\NW, \pi\in \QUERY(w)\}$.

The notion of a query automata can now be lifted from words
to nested words straightforwardly: a query automaton
for nested words over $\Sigma$ is a \SHA $A$
with alphabet $\Sigma\cup\{x,\neg x\}$. It
defines the unique total query $\QUERY$ such that
$L_\QUERY=\L(A)\cap\L(\onexnw)$.
\Long{A deterministic query automaton for the \XPath $\QNN7$ 
on \XML documents is given in \Figure{dSHA-QN7}.}

\section{Schema-Based Determinization for SHAs}
\label{sec:schema-det}

We can lift all previous algorithms from \NFAs to \SHAs
while extending the system of inference rules. The additional rules
concern tree initial states, that work in analogy
to initial states, and also apply rules that works similarly as 
internal rules. The new inference rules
for accessible products $A\times S$ are given in \Figure{extaccprod}
\Long{and for projection $\prf{A}(A\times S)$ in \Figure{extprojection}}. As before we define
$\scl{S}(A) = \prf{A}(A \times S)$. The rules for  schema-based
determinization $\detS(A)$ are extended in \Figure{detsha}. 
  The complexity upper bound, however, now becomes
  quadratic even with fixed alphabet:

\begin{figure}[t]
$$
  \infer[]
  {
(q,s) \in \Qinittgen{\Delta^{A\times S}} \quad 
(q,s) \in \states^{A\times S}
  }
  {
  q \in \Qinittgen{\Delta^A} & s \in \Qinittgen{\Delta^S}
  } 
\quad
 \infer[]
  {
          \SAar{(q_1,s_1)}{(q,s)}{(q_2,s_2)} \in \Delta^{A\times S} \quad (q_2,s_2) \in \states^{A\times S}
  }
       {
         \begin{array}{c}
           (q_1,s_1) \in \states^{A\times S} \\
           (q,s) \in \states^{A\times S}
         \end{array}
         \qquad
         \begin{array}{c}  \SAar{q_1}{q}{q_2} \in \Delta^A \\
                           \SAar{s_1}{s}{s_2} \in \Delta^S  
    \end{array}
  }
$$
\caption{\label{extaccprod} Lifting  accessible products to \SHAs.}

\end{figure}
\begin{toappendix}
\begin{figure}[t]
$$
\infer[]
  {
\begin{array}{c@{\qquad}c}
	q \in \Qinittgen{\Delta^{\prf{A}(A \times S)}} 
\end{array}
  }
  {
  	(q,s) \in \Qinittgen{\Delta^{A\times S}} 
  }
  \quad
   \infer[]
  {
  \SAar{q_1}{q}{q_2} \in \Delta^{\prf{A}(A \times S)}
  }
  {
   \SAar{(q_1,s_1)}{(q,s)}{(q_2,s_2)} \in \Delta^{A\times S} 
  }
$$
\caption{\label{extprojection} Lifting projections  $
\prf{A}(A\times S)$ to \SHAs.}
\end{figure}
\end{toappendix}
\begin{figure}[t]
$$
\infer[]
{
\Qinittgen{\Delta^A}\in \Qinittgen{\Delta^{\detS(A)}} \quad
\Qinittgen{\Delta^A}\sim s
}
{
\Qinittgen{\Delta^S}=\{s\}
}
\qquad
\infer[]
   {
Q_1@Q_2 \to  Q'\in\Delta^{\detS(A)} \qquad Q'\sim s'
    }
  {
\begin{array}{c}
s_1@s_2\to s' \in \rules^S \qquad Q_1\sim s_1\qquad Q_2\sim s_2 
  \\
 Q_1@Q_2 \to  Q' \in \Delta^{det(A)} 
\end{array}
    }
$$
\caption{\label{detsha} Extension of schema-based determinization to \SHAs.}
\end{figure}

\begin{propositionrep}\label{schema-clean-sha}
  If $A$ and $S$ are $\dSHAs$ then the accessible product
  $A\times S$ and the schema-based cleaning $\sclS(A)$ can be computed
  in expected amortized time $O(|\states^{A\times S}|^2 + |\states^{A \times S}| |\Sigma| + |A| + |S| )$.
\end{propositionrep}

\begin{toappendix}
 The algorithm in \Figure{algo-app-rules} is obtained by lifting the
 algorithm for \DFAs in \Figure{algo-int-rules-words} to \SHAs. For
 the case of apply rules, we have to combine each pair
 $(q_1,s_1)\in\states^{A\times S}$ in the stack $\mathit{Agenda}$ with
 all  $(q,s)\in\states^{A\times S}$ in the hash set $\mathit{Store}$,
 in both directions. The time to treat these pairs is $O(|\states^{A
   \times S}|^2)$, so quadratic in the worst case.
 As before, no state $(q_1,s_1)$ will be processed twice, due to the set membership test before pushing a pair into the agenda. 
\end{toappendix}

\begin{toappendix}
\begin{figure}[t]
\begin{sml}
fun $A \times S$ =
  let $Store=hashset.new(\emptyset)$ and $Agenda=list.new()$ and $Rules=hashset.new(\emptyset)$
  if $\init^A=\{q_0\}$ and  $\init^S=\{s_0\}$ then $Agenda.add((q_0,s_0))$  
  while $Agenda.notEmpty()$ do
    let $(q_1,s_1)= Agenda.pop()$
    for $a \in \Sigma$ do
      let $Q=\{ q_2 \mid \SAir{q_1}{a}{q_2} \in \Delta^A\}$ and $R=\{ s_2 \mid \SAir{s_1}{a}{s_2} \in \Delta^S\}$
      for $q_2 \in Q$ and $s_2 \in R$ do $Rules.add(\SAir{(q_1,s_1)}{a}{(q_2,s_2)})$
        if not $Store.member((q_2,s_2))$
           then $Store.add((q_2,s_2))$  $Agenda.push((q_2,s_2))$
    for $(q,s) \in Store$ do
      let $Q'=\{ q_2 \mid  \SAar{q_1}{q}{q_2} \in \Delta^A\}$ and $R'=\{ s_2 \mid  \SAar{s_1}{s}{s_2} \in \Delta^S\}$
      for $q_2 \in Q'$ and $ s_2 \in R'$ do $Rules.add(\SAir{(q_1,s_1)}{(q,s)}{(q_2,s_2)})$
        if not $Store.member((q_2,s_2))$
           then $Store.add((q_2,s_2))$  $Agenda.push((q_2,s_2))$
      let $Q''=\{q_2 \mid  \SAar{q}{q_1}{q_2} \in \Delta^A\}$ and $R''=\{ s_2 \mid  \SAar{s}{s_1}{s_2} \in \Delta^S\}$
      for $q_2 \in Q''$ and $ s_2 \in R''$ do $Rules.add(\SAir{(q,s)}{(q_1,s_1)}{(q_2,s_2)})$
        if not $Store.member((q_2,s_2))$
           then $Store.add((q_2,s_2))$  $Agenda.push((q_2,s_2))$
  let $\init^{A\times S}=\{(q_0,s_0) \mid (q_0,s_0) \in Store \}$ and $F^{A \times S}=\{(q,s) \mid (q,s) \in Store, q\in F^A, s \in F^S\}$
  return $(\Sigma,Store.toSet(),Rules.toSet(),\init^{A\times S},F^{A \times S})$
\end{sml}
\caption{\label{algo-app-rules} An algorithm computing the accessible product of \dSHAs $A$ and $S$.}
\end{figure}
\end{toappendix}

\begin{theo}[Correctness]\label{detS2} $\detS(A) = \schemaclean{\det(A)}{S}$
 for any \SHA $A$ and dSHA $S$
  with the same alphabet.
\end{theo}

\begin{propositionrep}\label{prop:sdet-sha}
The schema-based determinization $\detS(A)$ 
    of a \SHA $A$ with respect to a \dSHA $S$ can be
    computed in expected amortized time
$O(|\states^{\det(A)\times S}|^2 + |\states^{\det(A)\times S}|\MULTC
|\Sigma|+ |\states^{\detS(A)}|^2\MULTC |\Delta^A|  + |A| + |S|  )$.
\end{propositionrep}

The proof of Theorem \ref{detS2} \Long{presented in Section
\ref{sec:schema-clean-proof-sha}} extends
on that for \NFAs (Theorem \ref{detS1})
in a direct manner.
Proposition \ref{prop:sdet-sha} follows the result in Proposition \ref{compl-detS-words} with an additional quadratic factor in the size of states of the product $det(A) \times S$ and the states of the schema-based determinized automaton. This is always due to the apply rules of type $\states^3$.
\begin{appendixproof}
  Analogously to the case of \NFAs on words. 
      The algorithm in
      \Figure{algo-schema-det-nwords} computes the fixed point
      of the inference
      rules of schema-based determinization of \SHAs. As for \NFAs, it 
      stores untreated alignments on a stack $Agenda$ and processed
      alignments in a hash set $Store$. It also collects
      transition rules in a hash set $Rules$. New alignments can
      now be produced by the the inference rule for apply transitions:
      for each alignment $Q_1\sim s_1$ on the $Agenda$
      and $Q_2\sim s_2$ in the $Store$,  the algorithm
      computes the sets $\{s \mid \SAar{s_1}{s_2}{s} \in\Delta^S\}$ and $\{Q \mid \SAar{Q_1}{Q_2}{Q} \in \Delta^{det(A)}\}$
      and pushes all pairs $Q \sim s$ outside the $Store$ to the $Agenda$. There may be at most one such pair since $S$ and $det(A)$ are deterministic.
      We also have to consider the symmetric case where $Q_1\sim s_1$ on the store and $Q_2\sim s_2$ on the $Agenda$.
      Thus, it is in time $O(|\states^{\det(A)\times S}|^2 )$ which is quadratic in the worst case.
      Added to the latter, the cost of computing the transition of $det(A)$ on the fly which is in worst case $O(|\states^{\detS(A)}|^2 \MULTC|\Delta^A| + |A|)$. Therefore, having the whole algorithm running, including the time for computing the internal rules, in $O(|\states^{\det(A)\times S}|^2 + |\states^{\det(A)\times S}|\MULTC
|\Sigma|+ |\states^{\detS(A)}|^2\MULTC |\Delta^A|  + |A| + |S|  )$ . \ignore{We also note that if $Q\sim s$ is added to the store than $(Q,s)\in\states^{\det(A)\times S}$.} 
\end{appendixproof}
By Propositions \ref{comp-detacc-nwords} and
  \ref{schema-clean-sha}, computing $\sclS(det(A))$ by
  schema-based cleaning after accessible determinization needs
  time in $O(|\states^{\det(A)\times S}|^2 + |\states^{\det(A)\times S}|\MULTC
|\Sigma|+ |\states^{\det(A)}|^2\MULTC |\Delta^A|  + |A| + |S|  )$.
  This complexity bound is similar to that of
schema-based determinization from Proposition \ref{prop:sdet-sha}. 
Since $\states^{\detS(A)}\subseteq \states^{\det(A)}$, 
Proposition \ref{prop:sdet-sha} shows that the worst case time complexity of
schema-based determinization is never worse than for schema-based cleaning after determinization. 
\begin{toappendix}
\begin{figure}[t]
\begin{sml}
fun detS($A$,$S$) =
  let $Store=hashset.new(\emptyset)$
  let $Agenda=list.new()$ and $Rules=hashset.new(\emptyset)$
  if $\init^A\not=\emptyset$ and  $\init^S=\{s_0\}$ then $Agenda.add(\init^A\sim s_0)$ 
  while $Agenda.notEmpty()$ do
    let $(Q_1 \sim s_1)= Agenda.pop()$
    for $a \in \Sigma$ do
      let $P=\{ Q_2 \mid \SAir{Q_1}{a}{Q_2} \in \Delta^{det(A)}\}$ and $R=\{ s_2 \mid \SAir{s_1}{a}{s_2} \in \Delta^S\}$
      for $Q_2 \in P$ and $s_2 \in R$ do $Rules.add(\SAir{Q_1}{a}{Q_2})$
         if not $Store.member(Q_2 \sim s_2)$ 
            then $Store.add(Q_2 \sim s_2)$  $Agenda.push(Q_2 \sim s_2)$
    for $(Q \sim s) \in Store$ do         
      let $P'=\{Q_2 \mid  \SAar{Q_1}{Q}{Q_2} \in \Delta^{det(A)}\}$ and $R'=\{s_2 \mid  \SAar{s_1}{s}{s_2} \in \Delta^S\}$
      for $Q_2 \in P'$ and $s_2 \in R'$ do $Rules.add(\SAar{Q_1}{Q}{Q_2})$
         if not $Store.member(Q_2 \sim s_2)$ 
            then $Store.add(Q_2 \sim s_2)$  $Agenda.push(Q_2 \sim s_2)$      
      let $P''=\{Q_2 \mid  \SAar{Q}{Q_1}{Q_2} \in \Delta^{det(A)}\}$ and $R''=\{s_2 \mid  \SAar{s}{s_1}{s_2} \in \Delta^S\}$
      for $Q_2 \in P''$ and $s_2 \in R''$ do $Rules.add(\SAar{Q}{Q_1}{Q_2})$
         if not $Store.member(Q_2 \sim s_2)$
            then $Store.add(Q_2 \sim s_2)$  $Agenda.push(Q_2 \sim s_2)$
  let $\init^{\detS(A)}=\{Q \mid Q \sim s \in Store, Q \cap \init^A\not=\emptyset\}$ and $F^{\detS(A)}=\{Q \mid Q \sim s \in Store, Q \cap F^A \not= \emptyset\}$
  return $(\Sigma,Store.toSet(),Rules.toSet(),\init^{\detS(A)},F^{\detS(A)})$
\end{sml}
\caption{\label{algo-schema-det-nwords} An algorithm for schema-based determinization of an \SHA $A$ and a $d$\SHA schema $S$}
\end{figure}
\end{toappendix}

\newcommand\XMLDM{\ensuremath{\mathit{S}_{\mathit{XML}}}\xspace}
\newcommand\AN[1]{\ensuremath{\mathit{A}_{\QNN{#1}}}\xspace}

\ignore{****I think this is know irrelevant**** We implemented our algorithm for schema-based determinization in
Scala. A description and benchmarking results
can be found in~\cite{antonio}. 
We used symbolic \SHAs, so any $\Sigma$-rule is counted as $1$. Let
$\AN7$ be the \SHA for the example query \QNN7. It has $145$ states
and size $348$. Determinization of $\AN7$ based on schema $\onexnw $ and the
\XML data model \XMLDM produces a \dSHA with $74$ states of size $277$. 
The minimization of this $\dSHA  $ (see \cite{Sakho}) yields the \dSHA in
\Figure{dSHA-QN7} with $22$ states and size $144$.
In contrast, accessible determinization without any schema produces a
\dSHA with $10.003$ states and size $1.633.790$, for which
our Datalog implementation of schema-based cleaning runs out of
memory.

In this example an alternative solution
  can be obtained without schema-based determinization:
  one can compute $\scl{\XMLDM}(\det(\AN7\times \onexnw))$. This is
   since the determinization of the accessible product
  $\AN7\times \onexnw$ remains sufficiently small  with $74$
  states and size $269$. For general query automata $A$ and
  schemas $S$ and $S'$, however, rather than computing
  $\scl{S}(\det(A\times S'))$ it is better  to compute the
  equivalent automaton $\det_{S}(A\times S')$,  in order to avoid the cleaning of a
  potentially large intermediate result $\det(A\times S')$ with
  respect to schema $S$.}


\begin{toappendix}
\input{\OldLong /schema-clean-proof-sha}
\input{\OldLong /schema-det-proof-sha}
\end{toappendix}

\section{Experiments}
\label{sec:exps}

In this section, we present an experimental evaluation of the sizes of
the automata produced by the different  determinization methods. For this,
we consider a scalable family of \SHAs that is compiled from the
following scalable family of \XPath queries where $n$ and $m$ are
natural numbers. 
\begin{verbatim}
       (Qn.m)    //*[self::a0 or ... or self::an]
                 [descendant::*[self::b0 or ... or self::bm]]
\end{verbatim}
Query \texttt{Qn.m} selects all elements of an \XML document, that
are named by either of \texttt{a0}, $\ldots$, \texttt{an}
and have some descendant element named by either
of \texttt{b1}, $\ldots$, \texttt{bm}.  We compile
those \XPath queries to \SHAs based on the
compiler from \cite{Sakho}. As schema $S$, we chose the product of the
\dSHA $\onext$
with a \dSHA for the XML data model 
given in \Figure{xmlschema}.
Beside the concepts presented above, this \SHA also has typed
else rules. Actually, we use a richer class of SHAs in the experiments, which is converted back into the class of the paper when showing the results (except for else rules and typed else rules).
\begin{figure}[tb]
\centering
  \includegraphics[scale=0.25]{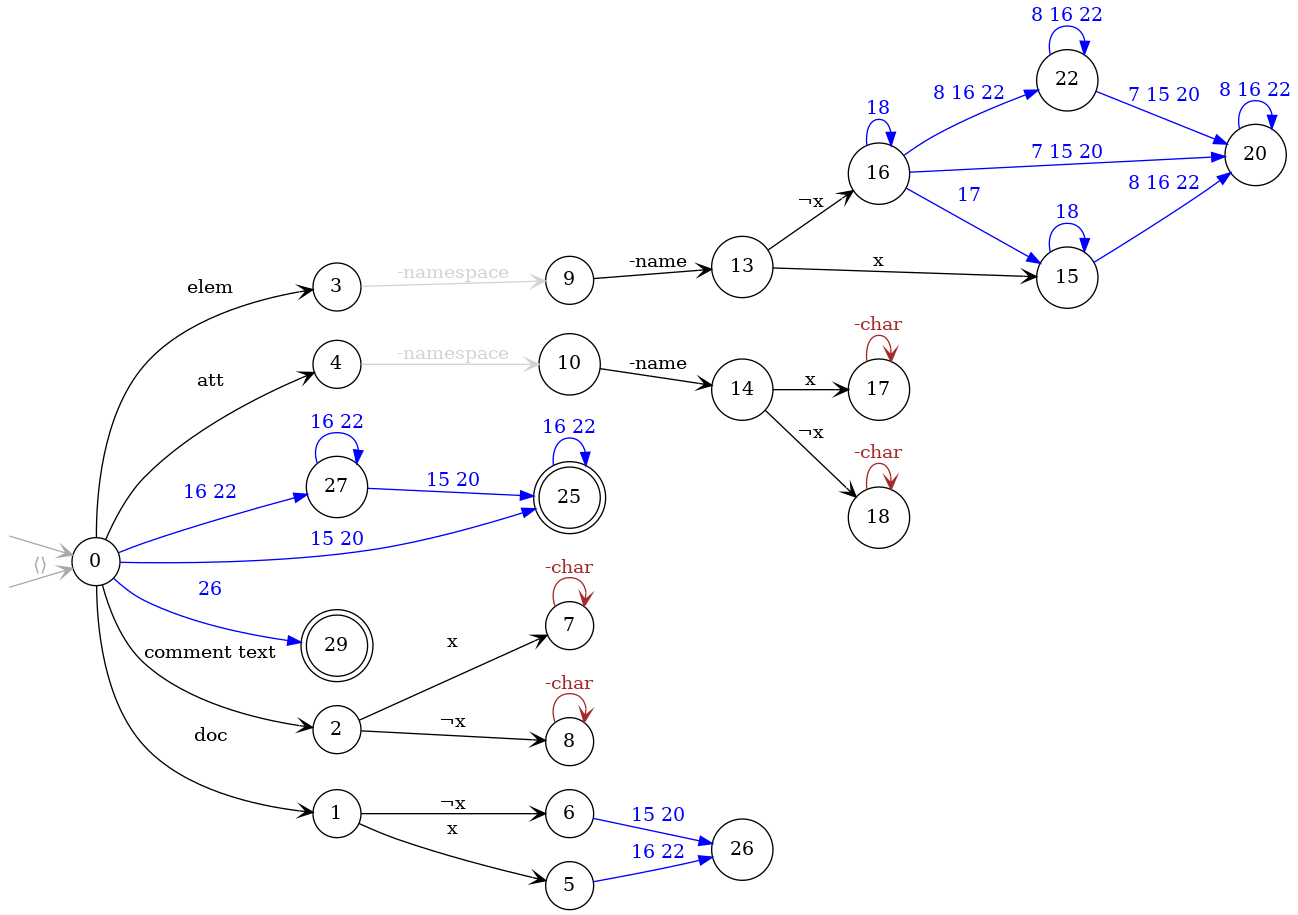}\vspace{-.2cm}
  \caption{\label{xmlschema}A schema for the intersection of \XML data model with $\onext$.}
\end{figure}

The results of our experiments are summarized in
  \Table{qnm}. For each
automaton we present two numbers,  size(\#states), its size and the
number of its states.
Unless specified otherwise, we use a timeout of 1000
seconds whenever calling some determinization algorithm.
Fields of the table are left blank if an exception was raised.
This happens when the determinization algorithm
reached the timeout, the memory was filled, or the stack overflowed.
We conducted all the experiments on a Dell laptop with the
following specs: Intel® Core™ i7-10875H CPU @ 2.30 GHz,16 cores, and
32 GB of RAM. 

The first column $A$ of \Table{qnm} reports on the \SHAs obtained from
  the queries Qn.m, by the compiler from \cite{Sakho} that is written in XSLT.
  The second column $det(A)$ is obtained from \SHA $A$ by accessible
  determinization. The blank cell in column $det(A)$ for query Q4.4
  was raised by a timeout of the  determinization algorithm. As one can see, this happens
  for all larger pairs ($n,m$). Furthermore, it appears that the sizes
  of the automata $det(A)$ grow
  exponentially with $n+m$.
  
  In the third column $det(A\times
  S)$, the determinization of the product is presented.
  It yields much smaller automata than with
  $det(A)$. For Q4.3 for instance, $det(A)$ has size 53550 (2161)
  while $det(A\times S)$ has size 5412 (438). The computation
  continues successfully until Q6.4. For the larger queries Q6.5 and
  Q6.6, our determinizer runs out of memory. 
  The fourth column $\detS(A)$  reports on schema-based
  determinization. For Q4.3 for instance we obtain
  3534 (329). Here and in all given examples, both measures are always
   smaller for $\detS(A)$ than for $det(A\times S)$. While this may not always
   be the case, but both approaches yield decent results generally.
   The numbers for the $\detS(A)$ for Q6.6 are marked in gray,
   since its computation took around one hour, so we obtain
   it only when ignoring the timeout. In contrast to $det(A\times
   S)$, however, the computation of $\detS(A)$ 
   did not run out of memory though. 
   The fifth column $\sclS(\det(A))$ contains the schema-based
   cleaning of $\det(A)$. This automaton is equal to $\detS(A)$
   by Correctness Theorem \ref{detS2}. Nevertheless, this
   cell is left blank in all but the smallest case $Q2.1$, since our datalog
   implementation of schema-based cleaning quickly runs out of
   memory for automata with many states.   The time in seconds that for determinization in $det(A\times S)$ and
   $\detS(A)$\INLINE{are given in \Table{qnmt}. It }\Long{are given in
     \Table{qnmt} of the appendix. It } 
   grows in dependence of the size of the output from 0.9 seconds
   until passing over the timeout\Long{ in the gray table cells}\INLINE{ in the gray table cells}.

  In the last two columns for $mini( det(A\times S))$ and $mini( detS(A))$
  we report the sizes of the minimization of $det(A\times S))$ and
  $\detS(A)$. It turns out that $mini(\detS(A))$
  is always smaller than $mini(det(A\times S))$, if both
  can be computed successfully. An example of  $mini(\detS(Q3.4))$ is shown in \Figure{msdet34}.

    \newcommand\prodS{A \times S}
  \begin{figure}[t]
  \centering
  \tiny
    \begin{tabular}{|l||*{7}{c|}}
      \hline
       & $A$& $det(A)$& $det(\prodS)$& $\detS(A)$&  $\sclS($& $mini(det$& $mini($\\
       &&&&&$det(A))$&$(\prodS))$&$\detS(A))$
     \\\hline\hline
      Q2.1&166 (67)&1380 (101)&540 (92)&284 (53)&284 (53)&160 (43)&73 (20) \\\hline\hline
      Q2.2&199 (79)&3635 (214)&1488 (167)&830 (106)&&162 (43)&75 (20) \\\hline\hline
      Q2.3&232 (91)&9574 (471)&4174 (334)&2424 (227)&&164 (43)&77 (20) \\\hline\hline
      Q2.4&265 (103)&24813 (1052)&11502 (713)&6826 (504)&&166 (43)&79 (20) \\\hline\hline
\INLINE{   Q3.1&203 (81)&3282 (204)&625 (104)&351 (64)&&162 (43)&75 (20) \\\hline\hline
      Q3.2&243 (95)&8660 (447)&1716 (191)&1025 (129)&&164 (43)&77 (20) \\\hline\hline
      Q3.3&283 (109)&22516 (996)&4793 (386)&2979 (278)&&166 (43)&79 (20) \\\hline\hline
      Q3.4&323 (123)&57328 (2225)&13148 (829)&8341 (619)&&168 (43)&81 (20) \\\hline\hline}
      Q4.1&240 (95)&8020 (435)&710 (116)&418 (75)&&164 (43)&77 (20) \\\hline\hline
      Q4.2&287 (111)&20945 (968)&1944 (215)&1220 (152)&&166 (43)&79 (20) \\\hline\hline
      Q4.3&334 (127)&53550 (2161)&5412 (438)&3534 (329)&&168 (43)&81 (20) \\\hline\hline
      Q4.4&381 (143)&&14794 (945)&9856 (734)&&170 (43)&83 (20) \\\hline\hline
\INLINE{      Q5.1&277 (109)&19722 (954)&795 (128)&485 (86)&&166 (43)&79 (20) \\\hline\hline
      Q5.2&331 (127)&50666 (2129)&2172 (239)&1415 (175)&&168 (43)&81 (20) \\\hline\hline
      Q5.3&385 (145)&&6031 (490)&4089 (380)&&170 (43)&83 (20) \\\hline\hline
      Q5.4&439 (163)&&16440 (1061)&11371 (849)&&&85 (20) \\\hline\hline}
      Q6.1&314 (123)&48212 (2113)&880 (140)&552 (97)&&168 (43)&81 (20) \\\hline\hline
      Q6.2&375 (143)&&2400 (263)&1610 (198)&&170 (43)&83 (20) \\\hline\hline
      Q6.3&436 (163)&&6650 (542)&4644 (431)&&172 (43)&85 (20) \\\hline\hline
      Q6.4&497 (183)&&18086 (1177)&12886 (964)&&&87 (20) \\\hline\hline
      Q6.5&558 (203)&&&34376 (2169)&&& \\\hline\hline
      Q6.6&619 (223)&&&{\color{gray}88666 (4862)}&&& \\\hline\hline 
	\end{tabular}
\caption{\label{qnm}Statistics of automata for \XPath queries: size(\#states)}
\end{figure}

 \begin{toappendix}
\begin{figure}[t]
\centering
\begin{tabular}{|l||*{5}{c|}}
\hline
&$det(\prodS)$&$time(det(\prodS))$&$\detS(A)$&$time(\detS(A))$\\\hline\hline
Q2.1&540 (92)&1.4&284 (53)&0.9\\\hline\hline
Q2.2&1488 (167)&3.1&830 (106)&1.5\\\hline\hline
Q2.3&4174 (334)&11.1&2424 (227)&3.9\\\hline\hline
Q2.4&11502 (713)&54.5&6826 (504)&16.9\\\hline\hline
Q3.1&625 (104)&1.6&351 (64)&1.1\\\hline\hline
Q3.2&1716 (191)&4.1&1025 (129)&1.9\\\hline\hline
Q3.3&4793 (386)&16.9&2979 (278)&6.2\\\hline\hline
Q3.4&13148 (829)&84.9&8341 (619)&30.6\\\hline\hline
Q4.1&710 (116)&1.9&418 (75)&1.3\\\hline\hline
Q4.2&1944 (215)&5.5&1220 (152)&2.7\\\hline\hline
Q4.3&5412 (438)&22.9&3534 (329)&9.3\\\hline\hline
Q4.4&14794 (945)&120&9856 (734)&46.5\\\hline\hline
Q5.1&795 (128)&2.2&485 (86)&1.5\\\hline\hline
Q5.2&2172 (239)&7.1&1415 (175)&3.3\\\hline\hline
Q5.3&6031 (490)&32.2&4089 (380)&13\\\hline\hline
Q5.4&16440 (1061)&174.6&11371 (849)&71.7\\\hline\hline
Q6.1&880 (140)&2.8&552 (97)&1.6\\\hline\hline
Q6.2&2400 (263)&9.2&1610 (198)&4.6\\\hline\hline
Q6.3&6650 (542)&44.33&4644 (431)&19\\\hline\hline
Q6.4&18086 (1177)&231.8&12886 (964)&101.4\\\hline\hline
Q6.5&&&34376 (2169)&569.7\\\hline\hline
Q6.6&&&{\color{gray}88666 (4862)}&3284.5\\\hline\hline
	\end{tabular}
        \caption{\label{qnmt}Timings in seconds for the determinization of the
          schema product and for schema-based determinization.}
\end{figure}
\end{toappendix}

\begin{figure}[tb]
\centering
  \includegraphics[scale=0.25]{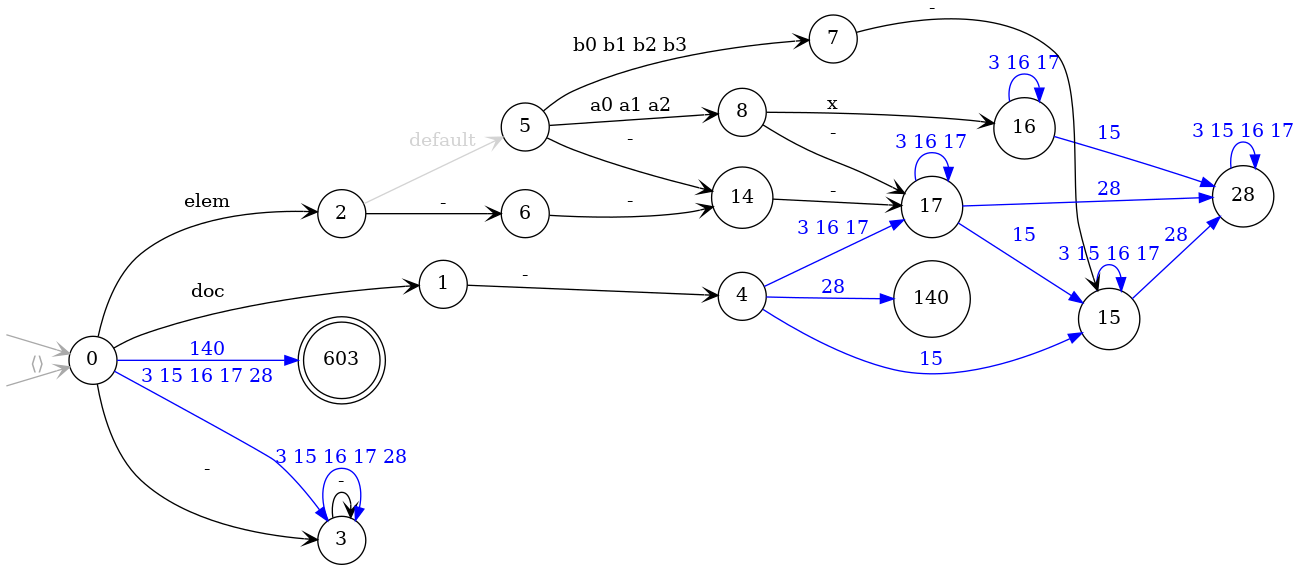}\vspace{-.2cm}
  \caption{\label{msdet34}The automaton $mini(\detS(A))$ of the query Q3.4.}
\end{figure}

\begin{toappendix}
\begin{figure}[tb]
\centering
  \includegraphics[scale=0.30]{Q34V2/mini/schemaprod.png}
  \caption{\label{msprod34}The automaton $mini(det(A \times S))$ of the query Q3.4.}
\end{figure}
\end{toappendix}
   

\subsection*{Conclusion and Future Work}
\label{sec:conclusion}
We presented an algorithm for schema-based determinization for
\SHAs and proved that it always produces the same results as
determinization followed by schema-based cleaning. We argued
why schema-based determinization is often way more efficient than
standard determinization, and why it is close in efficiency to the
determinization of the schema-product. The statements are
supported by upper complexity bounds and experimental evidence.
The experimental results of the present paper are enhanced by
follow up work \cite{alserhalibench}. They show that one can indeed
obtain small deterministic automata based on schema-based
determinization of stepwise hedge automata
for all regular \XPath queries in practice. We hope
that these automata are useful in the future for
experiments with query answering.


\bibliographystyle{eptcs}
\bibliography{\OldLong /core/mostrare,\OldLong /core/new,\OldLong /core/tom}


\end{document}